\documentclass{article}

\usepackage[papersize={8.5in,11in},margin=1in]{geometry}
\usepackage{comment}
\usepackage[colorlinks]{hyperref}
\usepackage{tabularx}
\usepackage{adjustbox}
\usepackage{graphicx}
\usepackage{algorithm}
\usepackage[noend]{algpseudocode}

\hypersetup{linkcolor=red,filecolor=red,citecolor=red,urlcolor=red}

\usepackage{amsmath}
\usepackage{amssymb}
\def\nats{{\mathbb N}}
\def\ints{{\mathbb Z}}
\usepackage{amsthm}
\newtheorem*{definition*}{Definition}
\newtheorem*{theorem*}{Theorem}
\newtheorem{theorem}{Theorem}[section]
\newtheorem{lemma}[theorem]{Lemma}

\newtheorem{observation}[theorem]{Observation}

\newtheorem{corollary}[theorem]{Corollary}

\newcommand{\idlow}[1]{\mathord{\mathcode`\-="702D\it #1\mathcode`\-="2200}}
\newcommand{\id}[1]{\ensuremath{\idlow{#1}}}

\newcommand{\namedref}[2]{\hyperref[#2]{#1~\ref*{#2}}}

\newcommand{\sectionref}[1]{\namedref{Section}{#1}}

\newcommand{\theoremref}[1]{\namedref{Theorem}{#1}}
\newcommand{\tableref}[1]{\namedref{Table}{#1}}
\newcommand{\lemmaref}[1]{\namedref{Lemma}{#1}}
\newcommand{\figureref}[1]{\namedref{Figure}{#1}}

\newcommand{\ellbw}{\ensuremath{\id{\ell-buffer-write}}}
\newcommand{\ellbr}{\ensuremath{\id{\ell-buffer-read}}}

\includecomment{fullversion}
\excludecomment{ignore}
\excludecomment{FAITH}

\begin{document}

\date{}

\title{A Complexity-Based Hierarchy for Multiprocessor Synchronization}

\author{
Faith Ellen\\
University of Toronto \\
\texttt{faith@cs.toronto.edu}
\and 
Rati Gelashvili\\
MIT\\
\texttt{gelash@mit.edu}
\and
Nir Shavit\\
MIT\\
\texttt{shanir@csail.mit.edu}
\and
Leqi Zhu\\
University of Toronto\\
\texttt{lezhu@cs.toronto.edu}
}

\maketitle
\begin{abstract}
For many years, Herlihy's elegant computability-based Consensus Hierarchy has been 
  our best explanation of the relative power of various types of multiprocessor 
  synchronization objects when used in deterministic algorithms.
However, key to this hierarchy is treating synchronization instructions as distinct objects, 
  an approach that is far from the real-world, where multiprocessor programs apply 
  synchronization instructions to collections of arbitrary memory locations. 
We were surprised to realize that, when considering instructions applied to memory locations, 
  the computability based hierarchy collapses. This leaves open the question of how to 
  better capture the power of various synchronization instructions. 

In this paper, we provide an approach to answering this question. 
We present a hierarchy of synchronization instructions, classified by the space complexity 
  necessary to solve consensus in an obstruction-free manner using these instructions. 
Our hierarchy provides a classification of combinations of known instructions that seems 
  to fit with our intuition of how useful some are in practice, 
  while questioning the effectiveness of others. 
In particular, we prove an essentially tight characterization of the power 
  of buffered read and write instructions.
Interestingly, we show a similar result for multi-location atomic assignments.
\end{abstract}


\section{Introduction}
Herlihy's Consensus Hierarchy \cite{Her91} assigns a consensus number to each object,
namely, the number of processes for which there is a wait-free binary consensus algorithm
using only instances of this object and read-write registers.
It is simple, elegant and,
for many years, has been our best explanation of synchronization power.

Robustness says that, using combinations of objects with consensus numbers at most $k$,
it is not possible to solve wait-free consensus for more than $k$ processes~\cite{Jay93}.
The implication is that modern machines need to provide objects with infinite consensus number.
Otherwise, they will not be universal, that is, they cannot be used to implement all objects or solve all tasks in a wait-free (or non-blocking) manner
for any number of processes~\cite{Her91,Tau06Book,Ray12Book,HS12Book}.
Although there are ingenious non-deterministic constructions that prove 
  that Herlihy's Consensus Hierarchy is not robust~\cite{Sch97,hl00}, 
  it is known to be robust for deterministic one-shot objects \cite{HR00}
and deterministic read-modify-write and readable objects \cite{Rup00}.
It is  unknown whether it is robust for general deterministic objects.

In adopting this explanation of computational power, we failed to notice an important fact:
multiprocessors do not compute using synchronization objects.
Rather, they apply synchronization instructions to locations in memory.
With this point of view, Herlihy's Consensus Hierarchy no longer captures
the phenomena we are trying to explain.

For example, consider two simple instructions:
\begin{itemize}
\item
$\id{fetch-and-add}(2)$, 
  which returns the number stored in a memory location and increases its value by 2, and
\item
$\id{test-and-set}()$, 
 which returns the number stored in a memory location
  and sets it to 1 if it contained 0.
\end{itemize}
(This definition of 
  $\id{test-and-set}$ is slightly stronger than the standard definition, 
  which always sets the location to which it is applied to 1.
Both definitions behave identically when the values in the location are in $\{0,1\}$.)
Objects that support only one of these instructions
have consensus number 2.
Moreover, these deterministic read-modify-write objects cannot be combined to solve 
  wait-free consensus for 3 or more processes. 
However,
with an object that supports both instructions, it is possible to solve 
  wait-free binary consensus for any number of processes.
The protocol uses a single memory location initialized to $0$.
Processes with input $0$ perform $\id{fetch-and-add}(2)$,
while processes with input $1$ perform $\id{test-and-set}()$.
If the value returned is odd, the process decides 1.
If the value 0 was returned from $\id{test-and-set}()$, the process also decides 1.
Otherwise,  the process decides 0.

Another example considers three instructions:
\begin{itemize}
\item
$\it{read}()$, which returns the number stored in a memory location,
\item
$\id{decrement}()$,
which decrements the number stored in a memory location and returns nothing, and
\item
$\id{multiply}(x)$,
which multiplies the number stored in a memory location by $x$ and returns 
nothing.
\end{itemize}
A similar situation arises: Objects that support only two of these instructions 
  have consensus number 1 and cannot be combined to solve wait-free consensus 
  for 2 or more processes. 
However, using an objects
that supports all three instructions,
  it is possible to solve wait-free binary consensus for any number of processes. 
The protocol uses a single memory location initialized to $1$. 
Processes with input $0$ perform $\id{decrement}()$, 
  while processes with input $1$ perform $\id{multiply}(n)$.
The second operation by each process is $\id{read}()$. 
If the value returned is positive, then the process decides 1. 
If it is negative, then the process decides 0.

For randomized computation, Herlihy's Consensus Hierarchy also collapses:
  randomized wait-free binary consensus among any number of processes 
  can be solved using only read-write registers, which have consensus number 1.
Ellen, Herlihy, and Shavit~\cite{FHS98} proved that $\Omega(\sqrt{n})$ 
historyless objects, 
which support only trivial operations, such as $\id{read}$, and historyless operations, such as $\id{write}$, $\id{test-and-set}$, and  $\id{swap}$,
are necessary to solve this problem. 
They noted that, in contrast, only one fetch-and-increment 
or fetch-and-add object suffices for solving this problem.
Yet, these objects and historyless objects are similarly classified in 
Herlihy's Consensus Hierarchy (i.e.~they all have consensus number 1 or 2). 
They suggested
that the number of instances
of an object needed to solve randomized wait-free consensus among $n$ 
processes might be another way to classify the power of the object.

Motivated by these observations, we consider a classification of 
  instruction sets based on the number of memory locations
  needed to solve \emph{obstruction-free $n$-valued} consensus among 
  $n \geq 2$ processes. 
Obstruction freedom is a simple and natural progress measure. 
Some state-of-the-art synchronization operations, 
  for example hardware transactions \cite{intel}, 
  do not guarantee more than obstruction freedom.
Obstruction freedom is also closely related to randomized computation.
In fact, any (deterministic) obstruction free algorithm can be transformed 
  into a randomized wait-free algorithm that 
  uses the same number of memory locations (against an oblivious adversary) 
  \cite{GHHW13}.
Obstruction-free algorithms can also be transformed into 
wait-free 
algorithms 
in the unknown-bound semi-synchronous model \cite{FLMS05}.
Recently, it has been shown that any
lower bound on the number of registers
used by obstruction-free algorithms also applies to  randomized wait-free algorithms~\cite{EGZ18}.

\subsection{Our Results}
Let {\em n-consensus} denote the problem of solving obstruction-free $n$-valued 
  consensus among $n \geq 2$ processes.
For any set of instructions $\mathcal{I}$,
let ${\cal SP}(\mathcal{I},n)$ denote the minimum number of memory locations 
supporting $\mathcal{I}$ that are needed to solve {\em n-consensus}.
This is a function from the positive integers, $\ints^+$, to $\ints^+ \cup \{\infty\}$.
For various instruction sets $\mathcal{I}$, we provide upper and lower bounds 
  on ${\cal SP}(\mathcal{I},n)$. 
The results are summarized in~\tableref{tab:hierarchy}.

We begin, in~\sectionref{sec:racing}, by considering the instructions
\begin{itemize}
\item 
$\id{multiply}(x)$,
which multiplies the number stored in a memory location by $x$ and returns 
nothing,
\item
$\id{add}(x)$,
which adds $x$ to the number stored in a memory location and returns nothing, and 
\item
$\id{set-bit}(x)$,
which sets bit $x$ of a memory location to 1 and returns nothing.
\end{itemize}
We show that one memory location supporting
$\id{read}()$ and one of these instructions
can be used to solve $n$-consensus.
The idea is to show that these instruction sets can implement $n$ counters in a 
single location. We can then use a \emph{racing counters} algorithm \cite{AH90}.

Next, we consider \emph{max-registers}~\cite{AAC09}. These  are 
memory locations supporting
\begin{itemize}
\item
$\id{read-max}()$,
which reads the number stored in a memory location, and
\item
$\id{write-max}(x)$,
which stores the number $x$ in a memory location, provided it contains a value 
less than $x$, and returns nothing.
\end{itemize}
In~\sectionref{sec:maxreg}, we prove that two max registers are necessary and 
sufficient for solving $n$-consensus.

In~\sectionref{sec:increment}, we prove that a 
single memory location supporting $\{ 
\id{read}(), 
\id{write}(x), 
\id{fetch-and-increment}() \}$ cannot be used to solve $n$-consensus, for $n 
\geq 3$.
We also present an algorithm for solving $n$-consensus using $O(\log n)$ 
such memory locations.
  
In~\sectionref{sec:buffer}, 
  we introduce a family of buffered read and buffered write instructions
  $\mathcal{B}_{\ell}$, for $\ell \geq 1$,
  and show how to solve $n$-consensus using $\lceil \frac{n}{\ell} \rceil$ 
  memory locations supporting these instructions. 
Extending Zhu's $n-1$ lower bound~\cite{Zhu16},
we also prove that $\lceil \frac{n-1}{\ell} \rceil$ such memory locations are 
necessary, which is tight except when $n-1$ is divisible by $\ell$.

Our main technical contribution is in~\sectionref{sec:transact}, where we show 
a lower bound of $\lceil \frac{n-1}{2\ell} \rceil$ locations,
even in the  presence of atomic multiple assignment.
Multiple assignment can be implemented by simple transactions,
  so our result implies that such transactions cannot significantly reduce 
  space complexity.
The proof further extends the techniques of~\cite{Zhu16} via a nice 
combinatorial argument, which
is of independent interest.

There are algorithms that solve $n$-consensus 
  using $n$ registers~\cite{AH90,BRS15,Zhu15}.
This is tight by the recent result of~\cite{EGZ18},
  which shows a lower bound of $n$ registers for binary consensus 
  among $n$ processes and, hence, for $n$-consensus.
  In~\sectionref{sec:swap}, we present a
modification of a known anonymous algorithm for $n$-consensus~\cite{Zhu15},
which solves $n$-consensus using 
$n-1$ memory locations supporting $\{ 
  \id{read}(), \id{swap}(x)\}$.
A lower bound of $\Omega(\sqrt{n})$ locations appears in~\cite{FHS98}.
This lower bound also applies to locations that only support
  $\id{test-and-set}()$, $\id{reset}()$ and $\id{read}()$ instructions.

Finally, in~\sectionref{sec:tas}, we show that an unbounded number of  
memory locations supporting $\id{read}()$ and either $\id{write(1)}$ or  
$\id{test-and-set}()$  
  are necessary and sufficient to solve $n$-consensus, for $n \geq 3$. 
Furthermore, we  show how to reduce the number of memory locations to $O(n\log n)$ when 
  in addition to $\id{read}()$, $\id{write}(0)$ and $\id{write}(1)$ are both available,
  or $\id{test-and-set}()$ and $\id{reset}()$ are both available.

  
\begin{table*}[ht]
\centering
\begin{tabular}{|c |c |}
\hline
Instructions $\mathcal{I}$ &  ${\cal SP}(\mathcal{I},n)$ \\
\hline
$\{ \id{read}(), \id{test-and-set}() \}$, $\{ \id{read}(), \id{write}(1) \}$ & $\infty$ \\
\hline
$\{ \id{read}(), \id{write}(1), \id{write}(0) \}$ & $n$ (lower), $O(n\log n)$ (upper) \\
\hline
$\{ \id{read}(), \id{write}(x) \}$ & $n$\\
\hline
$\{ \id{read}(),  \id{test-and-set}(),  \id{reset}() \}$ & $\Omega(\sqrt{n})$ (lower), $O(n\log n)$ (upper) \\
\hline
$\{ \id{read}(), \id{swap}(x) \}$ & $\Omega(\sqrt{n})$  (lower), $n-1$ (upper) 
\\
\hline
$\{ \id{\ell-buffer-read}(), \id{\ell-buffer-write}(x) \}$ & $\lceil \frac{n-1}{\ell} \rceil$ (lower), $\lceil \frac{n}{\ell} \rceil$ (upper) \\
\hline
 $\{ \id{read}(), \id{write}(x), \id{increment}() \}$ &2 (lower), $O(\log {n})$ 
 (upper)\\
$\{ \id{read}(), \id{write}(x), \id{fetch-and-increment}() \}$ & \\
\hline
$\{ \id{read-max}(), \id{write-max}(x) \}$ & 2 \\
\hline
$\{ \id{compare-and-swap}(x, y) \}$
\begin{fullversion}
$\{ \id{read}(), \id{set-bit}(x) \}$
\end{fullversion}
& 1\\ 
$\{ \id{read}(), \id{add}(x) \}$, $\{ \id{read}(), \id{multiply}(x) \}$ & \\
$\{ \id{fetch-and-add}(x) \} \}$, $\{ \id{fetch-and-multiply}(x) \}$ & \\
\hline
\end{tabular}
\caption{Space Hierarachy}
\label{tab:hierarchy}
\end{table*}
\section{Model}
\label{sec:model}
We consider an asynchronous system of $n \geq 2$ processes, with ids $0,1\ldots,n-1$,
  that supports a set of deterministic synchronization 
instructions, $\mathcal{I}$,
  on a set of identical memory locations.
The processes take steps at arbitrary, possibly changing, speeds and may crash 
at any time.
Each step is an atomic invocation of some instruction on some memory location by some process.
Scheduling is controlled by an adversary.
This is a standard asynchronous shared memory 
  model~\cite{AW04}, with the restriction that every memory location supports the same set of instructions.
We call this restriction the \emph{uniformity requirement}.

When allocated a step by the scheduler, a process performs one instruction on one shared memory
location and, based on the result, may then perform an arbitrary amount of local computation.
A {\em configuration} consists of the state of every process and the contents of every memory location.

Processes can use instructions on the memory locations
  to simulate (or implement) various objects. 
An object provides a set of operations which processes can call to access and/or change the value
of the object.
Although a memory location together with the supported instructions
can be viewed as an object, we do not do so,
to emphasize the uniformity requirement.
  

We consider the problem of solving 
  \emph{obstruction-free $m$-valued consensus} in such a system. 
Initially, each of the $n$ processes has an input from $\{0,1,\dots,m-1\}$ and 
  is supposed to output a value (called a \emph{decision}), such that all 
  decisions are the same (\emph{agreement}) and equal to the input of one of the processes (\emph{validity}).
Once a process has decided (i.e.~output its decision), the scheduler does not allocate it any further steps.
Obstruction-freedom means that, from each reachable configuration, each process
will eventually decide a value in a {\em solo execution}, i.e.~if the adversarial scheduler gives it sufficiently many consecutive steps.
When $m = n$, we call this problem \emph{$n$-consensus} and, when $m=2$, we call this problem 
{\em binary consensus}.
Note that lower bounds for binary consensus also apply to $n$-consensus.

In every reachable configuration of a consensus algorithm, each process has either decided
or has  one specific instruction it will perform on a particular memory location when next allocated a step by the scheduler. In this latter case, we say that the process is \emph{poised} to perform that instruction
on that memory location in the configuration.

  
\section{Arithmetic Instructions}
\label{sec:racing}
Consider a system that supports only $\id{read}()$ and either 
$\id{add}(x)$, $\id{multiply}(x)$, or $\id{set-bit}(x)$.
We show how to solve $n$-consensus using a single memory location in such a system. 
The idea is to show that we can simulate certain collections of objects that 
can solve $n$-consensus.

An \emph{$m$-component unbounded counter} object has $m$ components,
each with a nonnegative integral value.
It supports an $\id{increment()}$ operation on each component,
which increments the count stored in the component by $1$, 
and a $\id{scan}()$ operation, 
which returns the counts of all $m$ components.
In the next lemma, we present a \emph{racing counters} algorithm that bears 
some similarity to a consensus algorithm by Aspnes and Herlihy \cite{AH90}.

\begin{lemma}
    \label{lem:cntr}
    It is possible to solve obstruction-free $m$-valued consensus among 
    $n$ processes using an \\ $m$-component unbounded counter.
\end{lemma}
\begin{proof} 
We associate a separate component $c_v$ with each possible input value $v$.
 All components are initially $0$.
 Each process alternates between  \emph{promoting} a value (incrementing the component
 associated with that value) and performing a $\id{scan}$ of all $m$ components.
A process first promotes its input value.
After performing a $\id{scan}$, if it observes that the
count stored in
component, $c_v$, associated with some value $v$ 
is at least $n$ larger than
the counts stored in
all other components,
it returns the value $v$.
Otherwise, it promotes the value associated with a component containing the largest count
(breaking ties arbitrarily).

If some process returns the value $v$,
then each other process will increment some component at most once
before next performing a $\id{scan}$.
In each of those $\id{scan}$s, the count stored in $c_v$ will still be 
larger than the counts stored in all other components.
From then on, these processes will promote value $v$ and 
keep incrementing $c_v$.
Eventually, 
the count in 
component $c_v$ will be at least $n$ larger than
the counts in all other components,
and these processes will return $v$, ensuring agreement.

Obstruction-freedom follows because a process running on its own will continue to
increment the same component, which will eventually be $n$ larger than the counts in all other  components.
\end{proof}

In this protocol, the counts stored in the components may grow arbitrarily large. 
The next lemma shows that it is possible to avoid this problem, provided
each component also supports a $\id{decrement}()$ operation.
More formally, an \emph{$m$-component bounded counter} object has $m$ components,
where each component stores a count in $\{0, 1, \ldots, 3n-1\}$.
It supports both $\id{increment()}$ and $\id{decrement()}$
operations on each component, along with a $\id{scan}()$ operation, 
which returns the count stored in every component.
If a process ever attempts to increment a component that has count $3n-1$
or decrement a component that has count $0$, the object breaks 
(and every subsequent operation invocation returns $\bot$). 

\begin{lemma}
    \label{lem:Bcntr}
    It is possible to solve obstruction-free $m$-valued consensus among 
    $n$ processes using an $m$-component bounded counter.
\end{lemma}
\begin{proof}
    We modify the construction in~\lemmaref{lem:cntr} slightly 
    by changing what a process does when it wants to increment $c_v$ 
    to promote the value $v$. 
    Among the other components
    (i.e.~excluding $c_v$), let $c_u$ be one that stores the largest count.
    If $c_u  < n$, it increments $c_v$, as before.
    If $c_u \geq n$, then, instead of incrementing $c_v$, it decrements $c_u$.
    
    A component with value $0$ is never decremented. 
    This is because, after the last time some process observed that it stored a count greater than or equal to $n$,
 each process will decrement the component at most once before performing a $\id{scan}()$.
    Similarly, a component $c_v$ never becomes larger than $3n-1$:
    After the last time some process observed it to have count less than $2n$, 
    each process can increment $c_v$ at most once before performing a $\id{scan}()$.
    If $c_v \geq 2n$, then either the other components are less than $n$,
    in which case the process returns without incrementing $c_v$, 
    or the process decrements some other component, instead of incrementing $c_v$.
\end{proof}

In the following theorem, we show how to simulate unbounded and bounded 
counter objects.
\begin{theorem}
    It is possible to solve $n$-consensus using a single memory location that 
    supports only $\id{read}()$ and either $\id{multiply}(x)$, 
    $\id{add}(x)$, or $\id{set-bit}(x)$.
\end{theorem}
\begin{proof}
    We first give an obstruction-free implementation of an $n$-component unbounded counter object
    using a single location that supports
    $\id{read}()$ and $\id{multiply}(x)$.
    By~\lemmaref{lem:cntr}, this is sufficient for solving $n$-consensus. 
    The location is initialized with value $1$.
For each $v \in \{0, \ldots, n-1\}$,  let $p_v$ be the $(v+1)$'st prime number.
A process increments component $c_v$ by performing $\id{multiply}(p_v)$. 
A $\id{read}()$ instruction returns the value $x$ currently 
stored in the memory location.$\id{scan}$
This provides a $\id{scan}$ of all components:
component $c_v$ is the exponent of $p_v$ in the prime decomposition of $x$.
    
    A similar construction does not work using only 
    $\id{read}()$ and $\id{add}(x)$ instructions.
    For example, suppose one component is incremented by calling $\id{add}(a)$ 
    and another component is incremented by calling $\id{add}(b)$.
    Then, the value $ab$ can be obtained by incrementing the first component $b$ times 
    or incrementing the second component $a$ times.
    
    However, we can use a single memory location 
    that supports $\{\id{read}(), \id{add}(x)\}$
    to implement an $n$-component bounded counter.
    By~\lemmaref{lem:Bcntr}, this is sufficient for solving consensus.
We view the value stored in the location as a number written in base $3n$ and
interpret the $i$'th least significant digit of this number 
as the count of  component $c_{i-1}$.
The location is initialized with the value 0.
To increment $c_{i}$, a process performs $\id{add}((3n)^{i})$,
    to decrement $c_{i}$, it performs $\id{add}(-(3n)^{i})$
    and $\id{read}()$ provides a $\id{scan}$ of all $n$ components.
    
Finally, in systems supporting $\id{read}()$ and $\id{set-bit}(x)$, we can implement an $n$-component
unbounded counter by viewing the memory location as being partitioned into blocks, each consisting of $n^2$ bits.
Initially all bits are 0.
Each process locally stores the number of times it has incremented each component $c_v$.
To increment component $c_v$, process $i$ sets the $(vn+i)$'th bit in block $b+1$ to 1,
where $b$ is the number of times it has previously incremented component $c_v$.
It is possible to determine to current stored in each component
via a single $\id{read}()$: The count stored in component $c_v$ is simply the sum
of the number of times each process has incremented $c_v$.
\end{proof}
\section{Max-Registers}
\label{sec:maxreg}
A \emph{max-register} object~\cite{AAC09} supports two 
operations, $\id{write-max}(x)$ and $\id{read-max}()$.
The $\id{write-max}(x)$ operation sets the value of the max-register to $x$ if 
$x$ is larger than the current value and $\id{read-max}()$ returns the 
current value of the max-register (which is the largest amongst all values  previously written to it).
We show that two max-registers are necessary and sufficient for solving 
  $n$-consensus. 
\begin{theorem}
It is not possible to solve 
obstruction-free
  binary consensus for $n \geq 2$ processes using a single max-register.
\end{theorem}
\begin{proof}
Consider a solo terminating execution $\alpha$ of process $p$ 
  with input $0$ and a solo terminating execution $\beta$ 
  of process $q$ with input $1$.
We show how to interleave these two executions so that the resulting 
  execution is indistinguishable to both processes from their 
  respective solo executions.
Hence, both values will be returned, contradicting agreement.

To build the interleaved execution, run both processes until they 
  are first poised to perform $\id{write-max}$. 
Suppose $p$ is poised to perform $\id{write-max}(a)$ 
  and $q$ is poised to perform $\id{write-max}(b)$.
If $a \leq b$, let $p$ take steps until it is next poised 
  to perform $\id{write-max}$ or until the end of $\alpha$, 
  if it performs no more $\id{write-max}$ operations.
Otherwise, let $q$ take steps until it is next poised to perform 
  $\id{write-max}$ or until the end of $\beta$. 
Repeat this until one of the processes reaches the end of its execution
  and then let the other process finish.
\end{proof}
\begin{theorem}
\label{thm:racemaxr}
It is possible to solve $n$-consensus for any number of 
  processes using only two max-registers.
\end{theorem}
\begin{fullversion}
\begin{proof}
We describe a protocol for $n$-consensus using two max-registers, $m_1$ and $m_2$. 
Consider the lexicographic ordering $\prec$ on the set 
  $S = \nats \times \{0,\ldots,n-1\} = \{(r,x) : r \geq 0 \textrm{ and } x \in \{0,\dots,n-1\} \}$.
Let $y$  be a fixed prime that is larger than $n$. 
Note that, for $(r,x), (r',x') \in S$, $(r,x) \prec (r',x')$ 
  if and only if $(x+1)y^r < (x'+1)y^{r'}$. 
Thus, by identifying $(r,x) \in S$ with $(x+1)y^r$, we may assume that $m_1$ and $m_2$ 
  are max-registers defined on $S$ with respect to the lexicographic ordering $\prec$.

Since no operations decrease the value in a max-register, it is possible to 
  implement an obstruction-free \emph{scan} operation on $m_1$ and $m_2$ 
  using the double collect algorithm in~\cite{AADGMS93}:
  A process repeatedly collects the values in both locations (performing $\id{read-max}()$ on each location to obtain its value)
  until it observes two consecutive collects with the same values.

Initially, both $m_1$ and $m_2$ have value $(0,0)$. 
Each process alternately performs $\id{write-max}$ on one component and
takes a \emph{scan} of both components.
It begins by performing $\id{write-max}(0,x')$ to $m_1$, 
  where $x' \in \{0,\dots,n-1\}$ is its input value. 
If $m_1$ has value $(r+1,x)$ and $m_2$ has value $(r,x)$ in the \emph{scan},
 then it decides $x$ and terminates. 
If both $m_1$ and $m_2$ have value $(r,x)$ in the \emph{scan},
  then it performs $\id{write-max}((r+1,x)$ to $m_1$. 
Otherwise, it performs $\id{write-max}$ to $m_2$ with the value
  of $m_1$ in the \emph{scan}.

To obtain a contradiction suppose that there is an execution in which some 
  process $p$ decides value $x$ and another process $q$ decides value $x' \neq x$. 
Immediately before its decision, $p$ performed a \emph{scan} where 
  $m_1$ had value $(r+1,x)$ and $m_2$ had value $(r,x)$, for some $r \geq 0$.
Similarly, immediately before its decision, $q$ performed a \emph{scan} where $m_1$ 
  had value $(r'+1,x')$ and $m_2$ had value $(r',x')$, for some $r' \geq 0$. 
Without loss of generality, we may assume that $q$'s \emph{scan} occurs after $p$'s \emph{scan}.
In particular, $m_2$ had value $(r,x)$ before it had value $(r',x')$.
So, from the specification of a max-register, $(r,x) \preceq (r',x')$. 
Since $x' \neq x$, it follows that $(r,x) \prec (r',x')$. 

We show inductively, for $j = r',\ldots,0$, that some process performed a \emph{scan} 
  in which both $m_1$ and $m_2$ had value $(j,x')$. 
By assumption, $q$ performed a \emph{scan} where $m_1$ had value $(r'+1,x')$.
So, some process performed $\id{write-max}(r'+1,x')$ on $m_1$. 
From the algorithm, this process performed a \emph{scan} where $m_1$ and $m_2$ 
  both had value $(r',x')$.
Now suppose that $0 < j \leq r'$ and some process performed a \emph{scan} in which 
  both $m_1$ and $m_2$ had value $(j,x')$.  
So, some process performed $\id{write-max}(j,x')$ on $m_1$. 
From the algorithm, this process performed a \emph{scan} where $m_1$ and $m_2$ 
  both had value $(j-1,x')$.

Consider the smallest value of $j$ such that $(r,x) \prec (j,x')$. 
Note that $(r,x) \prec (r',x)$, so $j \leq r'$.
Hence, some process performed a \emph{scan} in which both $m_1$ and $m_2$ had value $(j,x')$.
Since $(r,x) \prec (j,x')$, this \emph{scan} occurred after the \emph{scan} by $p$, 
  in which $m_2$ had value $(r,x)$.
But $m_1$ had value $(j,x')$ in this \emph{scan} and $m_1$ had value $(r+1,x)$ 
  in $p$'s \emph{scan}, so $(r+1,x) \preceq (j,x')$. 
Since $x \neq x'$, it follows that $(r+1,x) \prec (j,x')$. 
Hence  $j \geq 1$  and $(r,x) \prec (j-1,x')$. 
This contradicts the choice of $j$.
\end{proof}
\end{fullversion}
\section{Increment}
\label{sec:increment}
Consider a system that supports only $\id{read}()$, $\id{write}(x)$, and 
$\id{fetch-and-increment}()$. We prove that it is not possible to solve $n$-consensus
using a single memory location.
We also consider a  weaker system that supports only $\id{read}()$, $\id{write}(x)$, and
$\id{increment}()$  and provide an algorithm using $O(\log n)$ memory locations.
\begin{theorem}
\label{thm:incnotone}
It is not possible to solve obstruction-free binary consensus 
  for $n \geq 2$ processes using a single memory location 
  that supports only $\id{read}()$, $\id{write}(x)$, and 
  $\id{fetch-and-increment}()$.
\end{theorem}
\begin{proof}
Suppose there is a binary consensus algorithm for two processes, 
  $p$ and $q$, using only one memory location. 
Consider solo terminating executions $\alpha$ and $\beta$ by 
  $p$ with input $0$ and input $1$, respectively. 
Let $\alpha'$ and $\beta'$ be the longest prefixes of $\alpha$ and $\beta$, respectively,
that do not contain 
  a $\id{write}$.
Without loss of generality, suppose that  at least 
  as many $\id{fetch-and-increment}()$ instructions are performed in $\beta'$
  as in $\alpha'$.
Let $C$ be the configuration that results from executing $\alpha'$ 
  starting from the initial configuration in which 
  $p$ has input $0$ and the other process, $q$ has input $1$.

Consider the shortest prefix $\beta''$ of $\beta'$ 
  in which $p$ performs the same number of $\id{fetch-and-increment}()$ instructions 
  as it performs in $\alpha'$. 
Let $C'$ be the configuration that results from executing $\beta''$ 
  starting from the initial configuration in which both 
  $p$ and $q$ have input $1$.
Then $q$ must decide $1$ in its solo terminating execution $\gamma$ 
  starting from configuration $C'$.
However, $C$ and $C'$ are indistinguishable to process $q$, 
  so it must decide $1$ in $\gamma$ starting from configuration $C$.
  If $p$ has decided in configuration $C$, then it has decided 0,
  since $q$ takes no steps in $\alpha'$.
Then both 0 and 1 are decided in execution $\alpha' \gamma$ starting from 
the initial configuration  in which 
  $p$ has input $0$ and  $q$ has input $1$.
  This violates agreement.
Thus, $p$ cannot have decided in configuration $C$.
 
Therefore, $p$ is poised to perform a $\id{write}$ in configuration $C$. 
Let $\alpha''$ be the remainder of $\alpha$, so 
  $\alpha = \alpha'\alpha''$.
Since there is only one memory location,
  the configurations resulting from performing this $\id{write}$ starting 
  from $C$ and $C\gamma$ are indistinguishable to $p$.
Thus, $p$ also decides $0$ starting from $C\gamma$.
But in this execution, both $0$ and $1$ are decided, violating agreement.
\end{proof}
The following well-known construction converts any algorithm for solving binary 
consensus to an algorithm for solving $n$-valued consensus~\cite{HS12Book}. 
\begin{lemma}
\label{lem:bitbybit}
Consider a system that supports a set of instructions that includes 
$\id{read}()$ and $\id{write}(x)$. If it is possible solve obstruction-free 
binary consensus among $n$ processes using only $c$ memory locations, then it 
is possible to solve $n$-consensus using only 
  $(c+2) \cdot \lceil \log_2 {n}\rceil - 2$ locations.
\end{lemma}
\begin{fullversion}
\begin{proof}
The processes agree bit-by-bit in $\lceil \log_2 {n}\rceil$ asynchronous rounds, 
  each using $c+2$ locations.
A process starts in the first round with its input value as its value for round $1$.
In round $i$, if the $i$'th bit of its value
  is $0$, a process writes its value in a designated $0$-location for the round. 
Otherwise, it writes its value in a designated $1$-location.
Then, it performs the obstruction-free binary consensus algorithm using $c$ locations
to agree on the $i$'th bit, $v_i$, of the output.
If this bit differs from the $i$'th bit of its value, 
  the process reads a recorded value from the designated $v_i$-location 
  for round $i$ and adopts its value for the next round. 
Note that some process must have already recorded a value to this
  location since, otherwise, the bit $\bar{v_i}$ would have been agreed upon.
  This ensures that the values used for round $i+1$ are all input values and
  they all agree in their first $i$ bits.
By the end, all processes have agreed on $\lceil \log_2 {n}\rceil$ bits,
  i.e.~on one of the at most $n$ different input values.scan

We can save two locations because the last round 
  does not require designated $0$ and $1$-locations.
\end{proof}
\end{fullversion}
We can implement a $2$-component unbounded counter,
  defined in~\sectionref{sec:racing},
  using two locations that support $\id{read}()$ and $\id{increment}()$.
The values in the two locations never decrease.
Therefore, as in the proof of~\theoremref{thm:racemaxr},
  a $\id{scan}()$ operation that returns the values of both counters
can be performed 
  using the double collect algorithm~\cite{AADGMS93}.
By~\lemmaref{lem:cntr}, $n$ processes can solve obstruction-free binary consensus
using a a $2$-component unbounded counter.
The next result then follows from~\lemmaref{lem:bitbybit}.
\begin{theorem}
It is possible to solve $n$-consensus using only $O(\log n)$ memory locations 
  that support only $\id{read}()$, $\id{write}(x)$, and $\id{increment}()$.
\end{theorem}
\section{Buffers}
\label{sec:buffer}
In this section, we consider the instructions $\id{\ell-buffer-read}()$ and  
  $\id{\ell-buffer-write}(x)$, for $\ell \geq 1$, which generalize read and write, 
  respectively. 
Specifically, an $\id{\ell-buffer-read}$  instruction returns the sequence of inputs 
  to the $\ell$ most recent $\id{\ell-buffer-write}$ instructions applied 
  to the memory location, in order from least recent to most recent. 
If the number of $\id{\ell-buffer-write}$ instructions previously applied to the memory
  location is $\ell' < \ell$, then the first $\ell - \ell'$ elements of this sequence are $\bot$.
Subsequent of the conference version of this paper~\cite{EGSZ16}, 
  Most{\'e}faoui, Perrin, and Raynal~\cite{MPR18} defined 
  a $k$-sliding window register, which is an object that supports only 
  $\id{k-buffer-read}$ and $\id{k-buffer-write}$ instructions.

We consider a system that supports the instruction set 
$\mathcal{B}_{\ell} = \{ \id{\ell-buffer-read}(), \id{\ell-buffer-write(x)} 
\}$, for some $\ell \geq 1$.
We call each memory location in such a system an $\id{\ell-buffer}$ and say
that each memory location has {\em capacity} $\ell$.
Note that a 1-buffer is simply a register. For $\ell > 1$, an 
$\ell$-buffer essentially maintains a buffer of the $\ell$ most recent writes 
to that location
and allows them to be read.

In~\sectionref{sec:buf-upper}, we show that a single $\ell$-buffer 
  can be used to simulate a powerful \emph{history} object that 
  can be updated by at most $\ell$ processes. 
This will allow us to simulate an obstruction-free variant of Aspnes 
  and Herlihy's algorithm for $n$-consensus~\cite{AH90} and, hence, 
  solve $n$-consensus, using only $\lceil n/\ell \rceil$ $\ell$-buffers.
In Section \ref{buf-lower}, we prove that $\lceil (n-1)/\ell \rceil$ 
$\ell$-buffers 
are necessary, which matches the upper bound whenever $n-1$ is not a multiple 
of $\ell$.
 
\subsection{Simulations Using Buffers}
\label{sec:buf-upper}

A \emph{history} object, $H$, supports two operations, $\id{get-history}()$ and 
$\id{append}(x)$, where $\id{get-history}()$ returns the sequence of all values 
appended to $H$ by prior $\id{append}$ operations, in order. We first show 
that, using a single $\ell$-buffer, $B$, we can simulate a history object, $H$,
that 
supports arbitrarily many readers and at most $\ell$ different appenders.

\begin{lemma}
\label{lem:historysim}
A single $\ell$-buffer can simulate a history object on which at most $\ell$ 
different processes
can perform $\id{append}()$ and any number of processes can perform 
$\id{get-history}()$.
\end{lemma}
\begin{proof}
Without loss of generality, assume that no value is appended to $H$
more than once. This can be achieved by having a process include
its process identifier and a sequence number along with the value
that it wants to append. 

In our implementation, $B$ is initially $\bot$ and each value written to
$B$ is of the form $(\mathbf{h},x)$, where $\mathbf{h}$ is a history of 
appended values and $x$ is a single
appended value.

To implement $\id{append}(x)$ on $H$, a process obtains a history, 
  $\mathbf{h}$, by performing $\id{get-history}()$ on $H$ and then 
  performs $\id{\ell-buffer-write}(\mathbf{h},x)$ on $B$. 
The operation is linearized at this $\id{\ell-buffer-write}$ step.

To implement $\id{get-history}()$ on $H$, a process simply performs an 
  $\id{\ell-buffer-read}$ of $B$ to obtain a vector $(a_1,\dots,a_{\ell})$,
  where $a_{\ell}$ is the most recently written value.
The operation is linearized at this $\id{\ell-buffer-read}$.
We describe how the return value of the $\id{get-history}()$ operation is computed.

We prove
that each $\id{get-history}()$ operation, $G$, on $H$ returns the sequence of inputs 
  to all $\id{append}$ operations on $H$ that were linearized before it, 
  in order from least recent to most recent.
Let $R$ be the $\id{\ell-buffer-read}$ step performed by $G$ and let $(a_1,\ldots,a_\ell)$ be the vector returned by $R$.

Note that   $(a_1,\ldots,a_\ell) = (\bot,\ldots,\bot)$ if and only if no
$\id{\ell-buffer-write}$ steps were performed before $R$
i.e.~if and only if no $\id{append}$ operations are linearized before $G$.
In this case, the empty sequence is returned by the $\id{get-history}()$ operation, as required.

Now suppose that $k \geq 1$ $\id{\ell-buffer-write}$ steps were performed on $B$ before $R$, i.e.~$k$
$\id{append}$ operations were linearized before $G$.
Inductively assume that each $\id{get-history}()$ operation
which has fewer than $k$ $\id{append}$ operations linearized before it
returns the sequence of inputs to those $\id{append}$ operations.

If $a_i \neq \bot$, then $a_i = (\mathbf{h}_i,x_i)$ was the input to an  $\id{\ell-buffer-write}$ step $W_i$ on $B$
performed before $R$. Consider the $\id{append}$ operation $A_i$ that performed step $W_i$. It appended the value $x_i$
to $H$ and the $\id{get-history}()$ operation, $G_i$, that $A_i$ performed returned the history
$\mathbf{h}_i$  of appended values .
Let $R_i$ be the $\id{\ell-buffer-read}$ step performed by $G_i$.
Since $R_i$ occurred before $W_i$, which occurred before $R$,
fewer than $k$ $\id{\ell-buffer-write}$ steps occurred before $R_i$.
Hence, fewer than $k$ $\id{append}$ operations are linearized before $G_i$.
By the induction hypothesis, $\mathbf{h}_i$ is the sequence of inputs to the 
$\id{append}$ operations linearized before $G_i$.

If $k < \ell$, then $a_1 = \cdots = a_{\ell -k} = \bot$.
In this case, $G$ returns the sequence $x_{\ell-k+1},\ldots,x_\ell$.
Since each $\id{append}$ operation is linearized at its 
  $\id{\ell-buffer-write}$ step and $x_{\ell-k+1},\ldots,x_\ell$ are the inputs
to these $k$ $\id{append}$ operations, 
  in order from least recent to most recent,
  $G$ returns the sequence of inputs to the $\id{append}$ operations
  linearized before it.
  
So, suppose that $k \geq \ell$.
 Let $\mathbf{h} = \mathbf{h}_m$ be the longest 
  history amongst $\mathbf{h}_1,\dots,\mathbf{h}_\ell$. 
If $\mathbf{h}$ contains $x_1$, then $G$
  returns $\mathbf{h}' , x_1,\ldots,x_\ell$, where $\mathbf{h}'$ 
  is the prefix of $\mathbf{h}$ up to, but not including, $x_1$.
By definition, $a_1, \ldots, a_\ell$ are the inputs to the last
$\id{\ell-buffer-write}$ operations prior to $R$,
so
$x_1,\dots,x_\ell$ are the last $\ell$ values appended
  to $H$ prior to $G$. 
Since $\mathbf{h}$ contains $x_1$, it 
  also contains all values appended to $H$ prior to $x_1$. 
It follows that $\mathbf{h}' \cdot (x_1,\dots,x_\ell)$ is the 
the sequence of inputs to the $\id{append}$ operations
  linearized before $G$.
  
\begin{figure*}[ht]
\centering
\begin{adjustbox}{max width = \textwidth}
{
\includegraphics[scale=1.35]{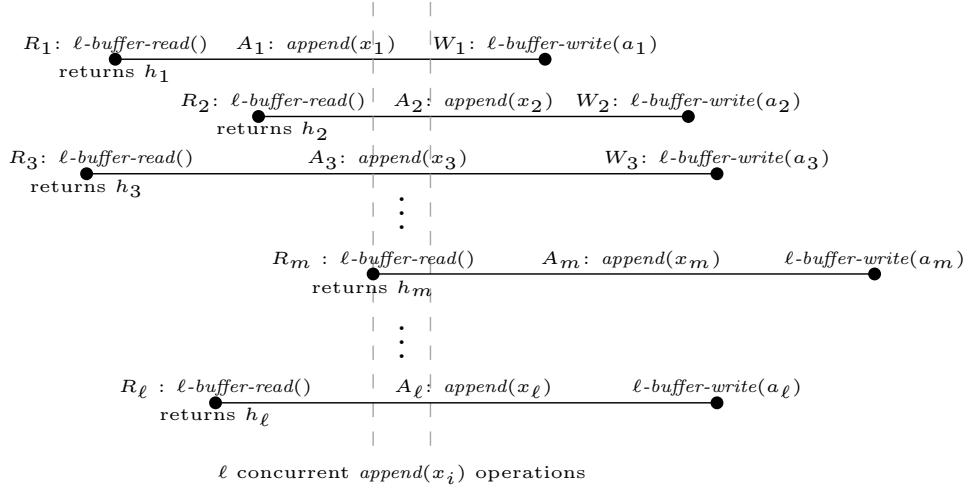}
}
\end{adjustbox}
\caption{When $\mathbf{h}$ does not contain $x_1$, there are $\ell$ concurrent $\id{append}$ operations.}
\label{fig:case2}
\end{figure*}

Now suppose that $\mathbf{h}$ does not contain $x_1$.
Then none of $\mathbf{h}_1, \ldots, \mathbf{h}_{\ell}$ contain $x_1$.
Hence $G_1, \ldots, G_{\ell}$ were linearized before $A_1$
and $R_1, \ldots,R_{\ell}$ were performed prior to $W_1$.
Since step $W_1$ occurred before $W_2,\ldots,W_{\ell}$,
the operations $A_1,\ldots,A_{\ell}$ are all concurrent with one another.
This is illustated in~\figureref{fig:case2}.
Therefore $A_1,\ldots,A_{\ell}$ are performed by different processes.
Only $\ell$ different processes can perform $\id{append}$ operations on $H$,
so no other $\id{append}$ operations on $H$ are  linearized between
$R_m$ and $W_1$.
Therefore, $\mathbf{h}$ contains all values appended to $H$ prior to $x_1$. 
It follows that $\mathbf{h} \cdot (x_1,\dots,x_\ell)$ 
  is the  sequence of inputs to the $\id{append}$ operations
  linearized before $G$.
\end{proof}

This lemma allows us to simulate \emph{any} object that supports 
  at most $\ell$ updating processes using only a single $\ell$-buffer. 
This is because the state of an object is determined by the history 
  of the non-trivial operations performed on it. 
In particular, we can simulate an array of $\ell$ single-writer 
  registers using a single $\ell$-buffer. 

\begin{lemma}
\label{lem:buffersim}
A single $\ell$-buffer can simulate $\ell$ single-writer registers.
\end{lemma}
\begin{fullversion}
\begin{proof}
Suppose that register $R_i$ is owned by process $p_i$, for $1 \leq i \leq \ell$.
By~\lemmaref{lem:historysim}, it is possible to simulate a history object $H$ that can
  be updated by $\ell$ processes and read by any number of processes.
To write value $x$ to $R_i$, process $p_i$ appends $(i,x)$ to $H$. 
To read $R_i$, a process reads $H$ and finds the value of the most recent write to $R_i$.
This is the second component of the last pair in the history whose first component is $i$.
\end{proof}
\end{fullversion}

Thus, we can use $\lceil \frac{n}{\ell} \rceil$ $\ell$-buffers 
  to simulate $n$ single-writer registers. 
An $n$-component unbounded counter shared by $n$ processes can be implemented 
  in an obstruction-free way from $n$ single-writer registers.
Each process records the number of times it has incremented each component
  in its single-writer register.
An obstruction-free $\id{scan}()$ can be performed using 
  the double collect algorithm~\cite{AADGMS93} and summing.
Hence, by~\lemmaref{lem:cntr} we get the following result.

\begin{theorem}
It is possible to solve $n$-consensus using only $\lceil n/\ell \rceil$ 
$\ell$-buffers.
\end{theorem}

\subsection{A Lower Bound}
\label{buf-lower}
In this section, we prove a lower bound on the number of 
  memory locations (supporting $\id{\ell-buffer-read}()$ and $\id{\ell-buffer-write}(x)$) necessary for solving obstruction-free 
  binary consensus among $n \geq 2$ processes.
  

In any configuration, location $r$ is \emph{covered} by process $p$ 
  if $p$ is poised to perform $\ellbw$ on $r$.
A location is \emph{$k$-covered} by a set of processes $\mathcal{P}$ in a configuration
if there are exactly $k$ processes in $\mathcal{P}$ that cover it.
A configuration is \emph{at most $k$-covered} by $\mathcal{P}$, if every process in $\mathcal{P}$ covers some location
and no location
  is $k'$-covered by $\mathcal{P}$,
  for any $k' > k$.

Let $C$ be a configuration and let $\mathcal{Q}$ be a set of processes, each of which is poised to perform
$\ellbw$ in $C$.
A \emph{block write} by $\mathcal{Q}$ \emph{from} $C$ is an execution,
  starting from $C$, in which each process in $\mathcal{Q}$ takes 
  exactly one step.
If a block write is performed that includes $\ell$ different
$\id{\ell-buffer-write}$ instructions to the same location,
and then some process performs $\id{\ell-buffer-read}$ on that location,
the process gets the same result
regardless of the value of that location
in $C$.


We say that a set of processes $\mathcal{P}$ \emph{can decide} 
  $v \in \{ 0, 1 \}$ \emph{from} a configuration $C$ if there exists 
  a $\mathcal{P}$-only execution from $C$ in which $v$ is decided.
If $\mathcal{P}$ can decide both $0$ and $1$ from $C$, 
  then $\mathcal{P}$ is \emph{bivalent} from $C$.

To obtain the lower bound, we extend the proof of the $n-1$ 
  lower bound on the number of registers required for solving 
  $n$-process consensus~\cite{Zhu16}.
We also borrow intuition 
  about reserving executions from the $\Omega(n)$ lower
  bound for anonymous consensus~\cite{Gel15}.
The following auxiliary lemmas are largely unchanged from~\cite{Zhu16}.
The main difference is that we only perform block writes on $\ell$-buffers
  that are $\ell$-covered by $\mathcal{P}$.
\begin{lemma}
  \label{lem:initbi}
There is an initial configuration from which the set of 
  all processes in the system is bivalent.
\end{lemma}
\begin{proof}
Consider an initial configuration, $I$, with two processes $p_0$ and $p_1$,
such that $p_v$ starts with input $v$,  for  $v \in \{0,1\}$.
Observe that $\{p_v\}$ can decide $v$ from $I$ since, 
  initially, $I$ is indistinguishable to $p_v$ from the configuration 
  where every process starts with input $v$.
Thus, $\{p_0, p_1\}$ is bivalent from $I$ and, therefore, 
  so is the set of all processes.
\end{proof}

\begin{lemma}
\label{lem:cover}
Let $C$ be a configuration and $\mathcal{Q}$ be a set processes
  that is bivalent from $C$.
Suppose $C$ is at most $\ell$-covered by a set of processes 
  $\mathcal{R}$, where $\mathcal{R} \cap \mathcal{Q} = \emptyset$.
Let $L$ be a set of locations that are $\ell$-covered 
  by $\mathcal{R}$ in $C$.
Let $\beta$ be a block write from $C$ by the set of 
  $\ell \cdot |L|$ processes in $\mathcal{R}$ that cover $L$.
Then there exists a $\mathcal{Q}$-only execution $\xi$ from $C$
  such that $\mathcal{R} \cup \mathcal{Q}$ is bivalent from 
 $C \xi \beta$ and, in configuration $C \xi$, some process 
  in $\mathcal{Q}$ covers a location not in $L$.
\end{lemma}
\begin{proof}
Suppose some process $p \in \mathcal{R}$ can decide some value 
  $v \in \{ 0, 1\}$ from configuration $C \beta$ and $\zeta$ is a 
  $\mathcal{Q}$-only execution from $C$ in which $\bar{v}$ is decided.
Let $\xi$ be the longest prefix of $\zeta$ such that 
  $p$ can decide $v$ from $C\xi\beta$.
Let $\delta$ be the next step by $q \in \mathcal{Q}$ 
  in $\zeta$ after $\xi$.

If $\delta$ is an $\ellbr$ or is an $\ellbw$ to a location in $L$,
  then $C\xi\beta$ and $C\xi\delta\beta$ 
  are indistinguishable to $p$.
Since $p$ can decide $v$ from $C\xi\beta$, 
  but $p$ can only decide $\bar{v}$ from $C\xi\delta\beta$, 
  $\delta$ must be an $\ellbw$ to a location not in $L$.
Thus, in configuration in $C\xi$, $q$ covers a location not in $L$
  and $C\xi\beta\delta$ is indistinguishable from 
  $C\xi\delta\beta$ to process $p$.
Therefore, by definition of $\xi$, 
  $p$ can only decide $\bar{v}$ from $C\xi\beta\delta$
  and $p$ can decide $v$ from $C\xi\beta$.
This implies that $\{p, q\}$ is bivalent 
  from $C\xi\beta$.
\end{proof}

The next result says that if a set of processes is bivalent in some
  configuration, then it is possible to reach a configuration 
  from which 0 and 1 can be decided
  in solo executions. 
It does not depend on what instructions are supported by the memory.
\begin{lemma}
\label{lem:twosolo}
Suppose $\mathcal{U}$ is a set of at least two processes that is bivalent 
  from configuration $C$.
Then it is possible to reach, via a $\mathcal{U}$-only execution from $C$, 
  a configuration, $C'$, such that, for $i = 0,1$, there is a process $q_i \in \mathcal{U}$
that decides $i$ from $C'$.
\end{lemma}
\begin{proof}
Let $\mathcal{D}$ be the set of all configurations 
  from which $\mathcal{U}$ is bivalent and which are reachable 
  from $C$ by a $\mathcal{U}$-only execution.
Let $k$ be the smallest integer such that there exist a configuration 
  $C'' \in \mathcal{D}$ and a set $\mathcal{U}' \subseteq \mathcal{U}$ of $k$ processes 
that is bivalent from $C''$.
Pick any such $C'' \in \mathcal{D}$ and let 
$\mathcal{U}' \subseteq \mathcal{U}$ be a set of $k$ processes that is bivalent from $C''$.
Since each process $p$ has only one terminating solo execution from $C''$
and it decides only one value in this execution, it follows that $k \geq 2$.

Consider a process $p \in \mathcal{U}'$ and 
  let $\mathcal{U}'' = \mathcal{U}' - \{p\}$ be 
  the set of remaining processes in $\mathcal{U}'$.
Since $|\mathcal{U}''| = k-1$, there exists $v \in \{0,1\}$ such that
  $\mathcal{U}''$ can only decide $v$ from $C''$.
  Let $q \in \mathcal{U}''$. Then $q$ decides $v$ from $C''$.
If $p$ decides $\bar{v}$ from $C''$, then $p$ and $q$ satisfy the claim
for $C' = C''$.
So, suppose that $p$ decides $v$ from $C''$.

Since $\mathcal{U}'$ is bivalent from $C''$, 
  there is a $\mathcal{U}'$-only execution 
  $\alpha$ from $C''$ that decides $\bar{v}$.
Let $\alpha'$ be the longest prefix of $\alpha$ such that both 
  $p$ and $\mathcal{U}''$ can only decide $v$ from $C''\alpha'$. 
Note that $\alpha' \neq \alpha$, because $\bar{v}$ is decided in $\alpha$.
Let $\delta$ be the next step in $\alpha$ after $\alpha'$.
Then either $p$ or $\mathcal{U}''$ can decide $\bar{v}$ 
  from $C''\alpha'\delta$.

First, suppose that $\delta$ is a step by a process in $\mathcal{U}''$.
Since $\mathcal{U}''$ can only decide $v$ from $C''\alpha'$, 
$\mathcal{U}''$ can only decide $v$ from $C''\alpha'\delta$.
Therefore, $p$ decides $\bar{v}$ from $C''\alpha'\delta$. 
Since  $q \in \mathcal{U}''$  decides $v$ from $C''\alpha'\delta$,
$p$ and $q$ satisfy the claim
for $C' = C''\alpha'\delta$.

Finally, suppose that $\delta$ is a step by $p$.
Since $p$ decides $v$ from $C''\alpha'$,
$p$ decides $v$ from $C''\alpha'\delta$.
Therefore, $\mathcal{U}''$ can decide $\bar{v}$ from $C''\alpha'\delta$.
However, $|\mathcal{U}''| = k-1$. 
By definition of $k$, 
$\mathcal{U}''$ is not bivalent from $C''\alpha'\delta$.
Therefore $\mathcal{U}''$ can only decide $\bar{v}$ from $C'\alpha'\delta$.
Since  $q \in \mathcal{U}''$ decides $\bar{v}$ from $C''\alpha'\delta$,
$p$ and $q$ satisfy the claim
for $C' = C''\alpha'\delta$.
\end{proof}

Similar to the induction used by Zhu~\cite{Zhu16},
from a configuration that is at most $\ell$-covered 
  by a set of processes $\mathcal{R}$,
we show how to reach another configuration that 
is at most $\ell$-covered 
  by $\mathcal{R}$ and in which another process 
  $z \not \in \mathcal{R}$ covers a location that is 
  not $\ell$-covered by $\mathcal{R}$. 
  
\begin{ignore}
 allows the processes 
  to reach configuration that is at most $\ell$-covered 
  by a set of processes $\mathcal{R}$, while another process 
  $z \not \in \mathcal{R}$ covers a location that is 
  not $\ell$-covered by $\mathcal{R}$. 
This implies that the configuration is also at most $\ell$-covered 
  by $\mathcal{R} \cup \{z\}$, allowing the inductive step to go through.
\end{ignore}
  
\begin{lemma}
\label{lem:induct}
Let $C$ be a configuration and let $\mathcal{P}$ be a set of 
  $n \geq 2$ processes.
If $\mathcal{P}$ is bivalent from $C$, then there is a 
  $\mathcal{P}$-only execution $\alpha$ starting from $C$ and a set $\mathcal{Q} \subseteq \mathcal{P}$
of two processes  such that $\mathcal{Q}$ 
  is bivalent from $C\alpha$ and $C\alpha$ is at most $\ell$-covered 
  by the remaining processes $\mathcal{P} - \mathcal{Q}$.
\end{lemma}
\begin{proof}
By induction on $|\mathcal{P}|$. 
The base case is when $|\mathcal{P}| = 2$. Let $\mathcal{Q} = \mathcal{P}$ and
let $\alpha$ be the empty execution.
Since $\mathcal{P} - \mathcal{Q}= \emptyset$,  the claim holds.
 
Now let $|\mathcal{P}| > 2$ and suppose the claim holds for 
  $|\mathcal{P}| - 1$. 
By~\lemmaref{lem:twosolo}, there exist a $\mathcal{P}$-only 
  execution $\gamma$ starting from $C$ and a set $\mathcal{Q} \subset \mathcal{P}$ of two processes 
   that is bivalent from $D = C\gamma$.
Pick any process $z \in \mathcal{P} - \mathcal{Q}$. 
Then $\mathcal{P} - \{z\}$ is bivalent from $D$
  because $\mathcal{Q}$ is bivalent from $D$.

We construct a sequence of configurations $D_0, D_1, \ldots$
  reachable from $D$ such that, for all $i \geq 0$, 
  the following properties hold:
\begin{itemize}
\item[1.] there exists a set of two processes 
  $\mathcal{Q}_i \subseteq \mathcal{P} - \{z\}$ such that
  $\mathcal{Q}_i$ is bivalent from $D_i$,
\item[2.] $D_i$ is at most $\ell$-covered by the remaining processes 
  $\mathcal{R}_i = (\mathcal{P} - \{z\}) - \mathcal{Q}_i$, and
\item[3.] if $L_i$ is the set of locations that are $\ell$-covered by $\mathcal{R}_i$ in $D_i$,
then $D_{i+1}$ is reachable from $D_i$ by a 
  $(\mathcal{P}-\{z\})$-only execution $\alpha_i$
  which contains a block write $\beta_i$ to $L_i$ by $\ell \cdot |L_i|$ processes in $\mathcal{R}_i$.
\end{itemize}

By  the induction hypothesis applied to $D$ 
and $\mathcal{P}-\{z\}$, there is a $(\mathcal{P}-\{z\})$-only execution $\eta$ starting from $D$
  and a set $\mathcal{Q}_0 \subseteq (\mathcal{P}-\{z\})$ of  two processes such that 
$\mathcal{Q}_0$ is bivalent from $D_0 = D\eta$ and $D_0$ is at most $\ell$-covered by 
  $\mathcal{R}_0 = (\mathcal{P} - \{z\}) - \mathcal{Q}_0$.

Now suppose that $D_i$ is a configuration reachable from $D$ and $\mathcal{Q}_i$ and
  $\mathcal{R}_i$ are sets of processes that satisfy all three conditions.
  
By~\lemmaref{lem:cover} applied to configuration $D_i$,
  there is a $\mathcal{Q}_i$-only execution $\xi_i$ such that 
  $\mathcal{R}_i \cup \mathcal{Q}_i = \mathcal{P} - \{z\}$
  is bivalent from $D_i\xi_i\beta_i$, where $\beta_i$ is a block 
  write to $L_i$ by $\ell \cdot |L_i|$ processes in $\mathcal{R}_i$.
Applying the induction hypothesis to $D_i\xi_i\beta_i$ 
  and $\mathcal{P} - \{z\}$, we get a $(\mathcal{P} - \{z\})$-only 
  execution $\psi_i$ leading to a configuration 
  $D_{i+1} = D_i\xi_i\beta_i\psi_i$, in which there is 
  a set, $\mathcal{Q}_{i+1}$, of two processes  such that 
  $\mathcal{Q}_{i+1}$ is bivalent from $D_{i+1}$.
Additionally, $D_{i+1}$ is at most $\ell$-covered by the set of remaining 
  processes $\mathcal{R}_{i+1} = (\mathcal{P} - \{z\}) - \mathcal{Q}_{i+1}$.
Note that
  the execution $\alpha_i = \xi_i \beta_i \psi_i$ contains
  the block write $\beta_i$ to $L_i$ by $\ell \cdot |L_i|$ processes in $\mathcal{R}_i$.

Since there are only finitely many locations, 
  there exists $0 \leq i < j$ such that $L_i = L_j$.
Next, we  insert steps of $z$ that cannot be detected by any process in 
  $\mathcal{P} - \{z\}$. 
Consider any $\{z\}$-only execution $\zeta$ from $D_i\xi_i$ 
  that decides a value $v \in \{0,1\}$. 
If $\zeta$ does not contain any $\ellbw$ to locations outside $L_i$, then
$D_i\xi_i\zeta\beta_i$ is indistinguishable 
  from $D_i\xi_i\beta_i$ to processes in $\mathcal{P} - \{z\}$.
Since $D_i\xi_i\beta_i$ is bivalent for $\mathcal{P} - \{z\}$,
there exists a  $\mathcal{P} - \{z\}$ execution from $D_i\xi_i\beta_i$
and, hence, from $D_i\xi_i\zeta\beta_i$
that decides  $\bar{v}$, contradicting agreement.
Thus $\zeta$  contains an $\ellbw$ to a location outside $L_i$.
Let $\zeta'$ be the longest prefix of $\zeta$ that does not contain
an  $\ellbw$ to a location outside $L_i$.
  Then,  in $D_i\xi_i\zeta'$, $z$ is poised to perform 
  an $\ellbw$ to a location outside $L_i = L_j$.

$D_i\xi_i\zeta'\beta_i$ is indistinguishable 
  from $D_i\xi_i\beta_i$ to $\mathcal{P} - \{z\}$, 
  so the $(\mathcal{P} - \{z\})$-only execution 
  $\psi_i \alpha_{i+1} \cdots \alpha_{j-1}$ can be applied from
$D_i\xi_i\zeta'\beta_i$. 
Let $\alpha = \gamma\eta\alpha_0 \cdots \alpha_{i-1} \xi_i\zeta'\beta_i\psi_i \alpha_{i+1} \cdots \alpha_{j-1}$. 
Every process in $\mathcal{P} - \{z\}$ is in the same state 
  in $C\alpha$ as it is in $D_j$. 
In particular, $\mathcal{Q}_j \subseteq \mathcal{P} - \{z\}$ 
  is bivalent from $D_j$ and, hence, from $C\alpha$. 
Every location is at most $\ell$-covered 
  by $\mathcal{R}_j = (P - \{z\}) - \mathcal{Q}_j$
  in $D_j$ and, hence, in $C\alpha$. 
Moreover, since $z$ takes no steps after $D_i\xi_i\zeta'$, 
  $z$ covers a location not  in $L_j$ in configurations 
$D_j$ and $C\alpha$. 
Therefore, every location is at most $\ell$-covered by 
  $\mathcal{R}_j \cup \{z\} = \mathcal{P} - \mathcal{Q}_j$ in $C\alpha$.
\end{proof}
Finally, we can prove the main theorem.
\begin{theorem}
\label{thm:nocons}
Consider a memory consisting of $\ell$-buffers.
Then any
obstruction-free binary consensus algorithm for $n$ processes
  uses at least $\lceil (n-1)/\ell \rceil$ locations.
\end{theorem}
\begin{proof}
Consider any obstruction-free binary consensus algorithm for $n$ processes.
By~\lemmaref{lem:initbi}, there exists
an initial configuration from which the set of all $n$ processes,
  $\mathcal{P}$, is bivalent. 
\lemmaref{lem:induct} implies that there is a configuration, $C$, 
  reachable from this initial configuration 
  and a set $\mathcal{Q} \subseteq \mathcal{P}$, of two processes  such
that $\mathcal{Q}$ is bivalent from $C$ and
$C$ is at most $\ell$-covered by the remaining processes
  $\mathcal{R} = \mathcal{P} - \mathcal{Q}$.
By the pigeonhole principle, $\mathcal{R}$ covers at least 
$\lceil (n-2)/\ell\rceil \geq \lceil (n-1)/\ell\rceil - 1$ different locations. 

Suppose that  $\mathcal{R}$ covers exactly $\lceil (n-2)/\ell\rceil$ different locations
and $\lceil (n-2)/\ell \rceil < \lceil (n-1)/\ell\rceil$.
Then $n-2$ is a multiple of $\ell$ and every location covered
by $\mathcal{R}$ is, in fact, $\ell$-covered by $\mathcal{R}$.
Since $\mathcal{Q}$ is bivalent from $C$, \lemmaref{lem:cover} implies that there is a $\mathcal{Q}$-only
execution $\xi$ such that some process in $\mathcal{Q}$ covers a location 
that is not covered by $\mathcal{R}$.
Hence, there are at least $\lceil (n-2)/\ell\rceil + 1 = \lceil (n-1)/\ell \rceil$  locations.
\end{proof}

The lower bound in \theoremref{thm:nocons} can 
  be extended to a {\em heterogeneous setting}, where the capacities of 
  different memory locations are not necessarily the same.
To do so,
we extend the definition of
a configuration $C$ being at most $\ell$-covered by a set of processes $\mathcal{P}$.
Instead, we require that the number of processes in $\mathcal{P}$ covering each location is at most its capacity.
Then we consider block writes to a set of locations containing $\ell$ different
  $\ellbw$ operations to each $\ell$-buffer in the set.
The general result is that, 
  for any algorithm which solves consensus for $n$ processes and
  satisfies nondeterministic solo termination,
  the sum of capacities of all buffers must be at least $n-1$.

The lower bound also applies to systems in which the return value 
  of every non-trivial instruction on a memory location does not depend 
  on the value of that location and the return value of any trivial instruction 
  is a function of the sequence of the preceding $\ell$ 
  non-trivial instructions performed on the location.
This is because such instructions can be implemented by $\id{\ell-buffer-read}$ and 
  $\id{\ell-buffer-write}$ instructions.
We record each invocation of a non-trivial instruction using $\id{\ell-buffer-write}$.
The return values of these instructions can be determined locally.
To implement a trivial instruction, we perform $\id{\ell-buffer-read}$, 
  which returns a sequence containing the description of the last $\ell$ 
  non-trivial instructions performed on the location. This is
  sufficient to determine the correct return value.
\section{Multiple Assignment}
\label{sec:transact}
With $m$-register multiple assignment, 
  we can atomically write to $m$ locations.
This instruction plays an important role in the Consensus Hierarchy~\cite{Her91}, 
  as $m$-register multiple assignment can used to solve wait-free 
  consensus for $2m-2$ processes, but not for $2m-1$ processes.

In this section, we explore whether multiple assignment could improve the
  space complexity of solving obstruction-free consensus.
A practical motivation for this question is that obstruction-free multiple
  assignment can be easily implemented using a simple transaction.

We prove a lower bound that is similar to the lower bound 
  in~\sectionref{buf-lower}.
Suppose $\id{\ell-buffer-read}()$ and $\id{\ell-buffer-write}(x)$ 
  instructions are supported on every memory location in a system and, 
  for any subset of locations, a process can atomically perform 
  one $\ellbw$ instruction per location.
Then $\lceil n/2 \ell \rceil$ locations are necessary for $n$ processes 
  to solve binary consensus.
As in~\sectionref{buf-lower}, this result can be further generalized to 
  a heterogeneous setting and different sets of instructions.

The main technical difficulty is proving an analogue
  of~\lemmaref{lem:cover}.
In the absence of multiple assignment, if $\beta$ is a block write
  to a set of $\ell$-covered locations, $L$, and $\delta$ is an $\ellbw$ 
  to a location not in $L$, then $\beta$ and $\delta$ commute 
  (in the sense that the configurations resulting from performing $\beta\delta$ and $\delta\beta$ are indistinguishable 
  to all processes).
However, a multiple assignment $\delta$ can atomically perform
  $\ellbw$ to many locations, including locations in $L$. 
Thus, it may be possible for processes to distinguish between 
  $\beta\delta$ and $\delta\beta$.
Using a careful combinatorial argument, we construct
  two blocks of multiple assignments, $\beta_1$ and $\beta_2$, such that, 
in each block,
  $\ellbw$ is performed at least $\ell$ times on each location in $L$
  and is not performed on any location outside of $L$. 
Given this, we can show that $\beta_1\delta\beta_2$ 
  and $\delta\beta_1\beta_2$ are indistinguishable to all processes. 
This is enough to prove an analogue of~\lemmaref{lem:cover}.

First, we define a notion of covering for this setting. 
In configuration $C$, process $p$ \emph{covers} location $r$ 
if $p$ is poised to perform a multiple assignment that includes an $\ellbw$ to $r$.
The next definition is key to our proof.
Suppose that, in some configuration $C$, each process in $\mathcal{P}$ is poised to perform a multiple assignment.
A $k$-\emph{packing} of $\mathcal{P}$  in $C$ 
  is a function $\pi$ mapping each process 
  in $\mathcal{P}$ to some memory location it covers such that 
  no location $r$ has more than $k$ processes mapped to it 
  (i.e., $|\pi^{-1}(r)| \leq k$).
When $\pi(p) = r$ we say that $\pi$ \emph{packs} $p$ in location $r$.
A $k$-packing may not always exist or there may be many $k$-packings, depending 
  on the configuration, the set of processes, and the value of $k$.
A location $r$ is \emph{fully $k$-packed} by $\mathcal{P}$ 
  in configuration $C$, if there is a $k$-packing of $\mathcal{P}$ in $C$ 
  and all $k$-packings  of $\mathcal{P}$ in $C$ pack exactly $k$ processes in $r$.

Suppose that, in some configuration, there are two $k$-packings 
  of the same set of processes, but the first packs more processes 
  in some location $r$ than the second.
We show there is a location $r'$ in which the first packing packs 
  fewer processes than the second and there is a $k$-packing which,
  as compared to the first packing, packs 
  one less process in location $r$, one more process in location $r'$,
  and the same number of processes in all other locations.
The proof relies on existence of a certain Eulerian path in a 
  multigraph that we build to represent these two $k$-packings.
  
\begin{lemma}
\label{lem:euler}
Suppose $g$ and $h$ are two $k$-packings of the same set of processes 
  $\mathcal{P}$ in some configuration $C$ and $r_1$ is a location 
  such that $|g^{-1}(r_1)| > |h^{-1}(r_1)|$
(i.e., $g$ packs more processes in $r_1$ than $h$ does).
Then, there exists a sequence of locations, $r_1, r_2, \ldots, r_t$, and 
  a sequence of distinct processes, $p_1, p_2, \ldots, p_{t-1}$, such that
$|h^{-1}(r_t)| > |g^{-1}(r_t)|$
(i.e., $h$ packs more processes in $r_t$ than $g$),
and  $g(p_i) = r_i$ and $h(p_i) = r_{i+1}$ for $1 \leq i \leq t-1$.
Moreover, for $1 \leq j <t$, there exists a $k$-packing $g'$ such that
$g'$ packs one less process than $g$ in $r_j$, $g'$ packs one more process than $g$ in $r_t$, 
$g'$ packs the same number of processes as $g$ in all other locations,
and $g'(q) = g(q)$ for all $q \not\in \{ p_j, \ldots, p_{t-1}\}$.
\end{lemma}
\begin{proof}
Consider a  multigraph with one node for each memory location 
  in the system and one directed edge from node $g(p)$ 
  to node $h(p)$ labelled by $p$, for each process $p \in \mathcal{P}$.
The in-degree of any node $v$ is $|h^{-1}(v)|$, which 
  is the number of processes that are packed into memory location $v$ 
  by $h$, and the out-degree of node $v$ is $|g^{-1}(v)|$, which 
  is the number of processes that are packed in $v$ by $g$.

Now, consider any maximal Eulerian path in this multigraph starting 
  from the node $r_1$.
This path consists of a sequence of distinct edges, but may visit the same node multiple times.
Let $r_1,\ldots,r_t$ be the sequence of nodes visited 
  and let $p_i$ be the labels of the 
  traversed edges, in order.
Then $g(p_i) = r_i$ and $h(p_i) = r_{i+1}$ for $1 \leq i \leq t-1$.
The edges in the path are all different and each is labelled by a different process,
  so the path has length at most $|\mathcal{P}|$.
By maximality, the last node in the sequence must have 
  more incoming edges than outgoing edges, so  $|h^{-1}(r_t)| > |g^{-1}(r_t)|$.

Let $1 \leq j < t$.
We construct $g'$ from $g$ by re-packing each process $p_i$ from $r_i$ 
  to $r_{i+1}$ for all $j \leq i < t$.
Then $g'(p_i) = r_{i+1}$ for $j \leq i < t$ and $g'(p) = g(p)$ 
  for all other processes $p$.
Notice that $p_i$ covers $r_{i+1}$, since $h(p_i) = r_{i+1}$ and 
  $h$ is a $k$-packing.
As compared to  $g$, $g'$ packs one less process in $r_j$, 
  one more process in $r_t$, and the same number of processes 
  in every other location.
Since $h$ is a $k$-packing, it packs at most $k$ processes in $r_t$.
Because $g$ is a $k$-packing that packs less processes in $r_t$ than $h$,
  $g'$ is also a $k$-packing.
\end{proof}

\begin{ignore}
\begin{corollary}
\label{cor:repack}
Let the $k$-packings $g$ and $h$ and the sequences $r_i$ and $p_i$ 
  be defined as in~\lemmaref{lem:euler}.
For $1 \leq j <t$, there exists a $k$-packing $g'$, such that $g'$ packs 
  one less process than $g$ in $r_j$, one more process than $g$ in $r_t$, 
  and the same number of processes as $g$ in all other locations. 
\end{corollary}
\begin{proof}
We construct $g'$ from $g$ by re-packing each process $p_i$ from $r_i$ 
  to $r_{i+1}$ for all $j \leq i < t$.
Then $g'(p_i) = r_{i+1}$ for $j \leq i < t$ and $g'(p) = g(p)$ 
  for all other processes $p$.
Notice that $p_i$ covers $r_{i+1}$, since $h(p_i) = r_{i+1}$ and 
  $h$ is a $k$-packing.

As compared to  $g$, $g'$ packs one less process in $r_j$, 
  one more process in $r_t$, and the same number of processes 
  in every other location.
Since $h$ is a $k$-packing, it packs at most $k$ processes in $r_t$.
Because $g$ is a $k$-packing that packs less processes in $r_t$ than $h$,
  $g'$ is also a $k$-packing.
\end{proof}
\end{ignore}

Let $\mathcal{P}$ be a set of processes, each of which is poised 
  to perform a multiple assignment in some configuration $C$.
A \emph{block multi-assignment} by $\mathcal{P}$ from $C$ is an execution 
  starting at $C$, in which each process in $\mathcal{P}$ 
  takes exactly one step.

Consider some configuration $C$ and a set of processes $\mathcal{R}$
  such that there is a $2 \ell$-packing $\pi$ of $\mathcal{R}$ in $C$.
Let $L$ be the set of all locations that are fully $2 \ell$-packed 
  by $\mathcal{R}$ in $C$, so $\pi$ packs exactly $2 \ell$ processes 
  from $\mathcal{R}$ in each location $r \in L$.
Partition the $2 \ell \cdot |L|$ processes packed by $\pi$ in $L$ 
   into two sets,
  $\mathcal{R}^1$ and $\mathcal{R}^2$,
each containing 
  $\ell \cdot |L|$ processes, such that, for each location $r \in L$,
$\ell$ of the processes packed in $r$ by $\pi$ belong to $\mathcal{R}^1$
and the other $\ell$ belong to $\mathcal{R}^2$.
For $i \in \{1,2\}$, let $\beta_i$ be a block multi-assignment by $\mathcal{R}^i$. 
  
Notice that, for any location $r \in L$,  
  the outcome of any $\ellbr$ on $r$ after $\beta_i$ does not depend on 
  multiple assignments that occurred prior to  $\beta_i$.
Moreover, we can prove the following crucial property 
  about these block multi-assignments to fully packed locations.
\begin{lemma}
\label{lem:satnotout}
Neither $\beta_1$ nor $\beta_2$ involves an $\ellbw$ to a location 
  outside of $L$.
\end{lemma}
\begin{proof}
Assume the contrary. Let $q \in \mathcal{R}^1 \cup \mathcal{R}^2$ 
  be a process with $\pi(q) \in L$ such that, in $C$, $q$ also covers some 
  location $r_1 \not \in L$.
If $|\pi^{-1}(r_1)| < 2 \ell$, then there is another $2\ell$ packing of 
$\mathcal{R}$ in $C$, which is the same as $\pi$, except that it packs $q$ in location $r_1$
instead of $\pi(q)$. However, this packing packs fewer than $2 \ell$ processes 
in $\pi(q) \in L$, contradicting the definition of $L$.
Therefore $|\pi^{-1}(r_1)| = 2 \ell$,   i.e.,
$\pi$ packs exactly $2 \ell$ processes in $r_1$.

Since $L$ is the set of all fully $2 \ell$-packed locations,
  there exists a $2 \ell$-packing $h$, which packs strictly fewer than 
  $2 \ell$ processes in $r_1 \not \in L$.
From~\lemmaref{lem:euler} with $g = \pi$ and $k = 2\ell$,
there is a sequence of locations, $r_1, \ldots, r_t$, and 
  a sequence of processes, $p_1, \ldots, p_{t-1}$,
such that $|h^{-1}(r_t)| > |\pi^{-1}(r_t)|$.
Since $h$ is a $2\ell$-packing, it packs at most $2\ell$ processes in $r_t$
and, hence, $\pi$ packs strictly less than $2\ell$ processes in $r_t$.
Thus, $r_t \not\in L$.
We consider two cases.

First, suppose that $q \neq p_i$ for all $i = 1,\ldots,t-1$, i.e.,
$q$ does not occur in the sequence $p_1, \ldots, p_{t-1}$.
By the second part of \lemmaref{lem:euler} with $j = 1$,
there is a $2\ell$-packing $\pi'$ that packs less than $2 \ell$ 
processes in $r_1$, one more process than $\pi$ in $r_t$,
and the same number of processes as $\pi$ in all other locations.
In particular, $\pi'$ packs exactly $2\ell$ processes in each location in $L$, including $\pi(q)$.
Moreover, $\pi'(q) = \pi(q)$,
since $q$ does not occur in the sequence $p_1, \ldots, p_{t-1}$.
Consider another $2\ell$ packing of 
$\mathcal{R}$ in $C$, which is the same as $\pi'$, except that it packs $q$ in location $r_1$
instead of location $\pi(q)$. However, this packing packs fewer than $2 \ell$ processes 
in $\pi(q) \in L$, contradicting the definition of $L$.

Now, suppose that $q = p_s$, for some $s \in \{1,\ldots, t-1\}$.
Since $r_s = \pi(p_s) = \pi(q) \in L$, it follows that $|\pi^{-1}(r_s)| = 2\ell$.
By the second part of \lemmaref{lem:euler} with $j = s$,
there is a $2\ell$-packing that packs less than $2 \ell$ 
processes in $r_s$, one more process than $\pi$ in $r_t$,
and the same number of processes as $\pi$ in all other locations.
Since $r_s \in L$, this contradicts the definition of $L$.

Thus, in configuration $C$, every process in 
$\mathcal{R}^1 \cup \mathcal{R}^2$ only covers locations in $L$.
\end{proof}
We can now prove a lemma that replaces~\lemmaref{lem:cover} in the main argument.

\begin{lemma}
\label{lem:multcover}
Let $\mathcal{Q}$ be a set of processes disjoint from $\mathcal{R}$ that is bivalent from $C$.
Then there exists a $\mathcal{Q}$-only execution $\xi$ from
  $C$ such that $\mathcal{R} \cup \mathcal{Q}$ is bivalent from $C \xi \beta_1$
  and, in configuration $C \xi$, some process in $\mathcal{Q}$ covers
  a location not in $L$.
\end{lemma}
\begin{proof}
Suppose some process $p \in \mathcal{R}$ can decide some value $v \in \{0, 1\}$
  from configuration $C \beta_1 \beta_2$ and $\zeta$ is a $\mathcal{Q}$-only execution
  from $C$ in which $\bar{v}$ is decided.
Let $\xi$ be the longest prefix of $\zeta$ such that $p$ can decide
  $v$ from $C \zeta \beta_1 \beta_2$.
Let $\delta$ be the next step by $q \in \mathcal{Q}$ in $\zeta$ after $\xi$.

If $\delta$ is an $\ellbr$ or a  multiple assignment involving only $\ellbw$ operations
  to locations in $L$,
  then $C\xi\beta$ and $C\xi\delta\beta$ 
  are indistinguishable to $p$.
Since $p$ can decide $v$ from $C\xi \beta_1 \beta_2$, but $p$ can only 
  decide $\bar{v}$ from $C\xi \delta \beta_1 \beta_2$, $\delta$ must be 
  a multiple assignment that includes an $\ellbw$ to a location not in $L$.
Thus, in configuration $C\xi$, $q$ covers a location not in $L$.
For each location $r \in L$, the value of $r$ is the
same in $C\xi\delta\beta_1\beta_2$ as it is in 
  $C\xi\beta_1\delta\beta_2$ due to the block multi-assignment $\beta_2$.
By~\lemmaref{lem:satnotout},  for each location $r \not\in L$, 
neither $\beta_1$ nor $\beta_2$ performs an $\ellbw$ to $r$, so the
value of $r$ is the same in $C\xi\delta\beta_1\beta_2$ as it is in
$C\xi\beta_1\delta\beta_2$.
Since the state of process $p$ is the same in configuration $C\xi\beta_1\delta\beta_2$ and $C\xi\delta\beta_1\beta_2$,
these two configurations are  indistinguishable to $p$.

Therefore, by definition of $\xi$, $p$ can only decide $\bar{v}$ from $C\xi\beta_1\delta\beta_2$
and $p$ can decide $v$ from $C\xi\beta_1\beta_2$.
This implies that $\mathcal{R} \cup \mathcal{Q}$ is bivalent from $C \xi \beta_1$.
\end{proof}
Using these tools, we can prove the following analogue of~\lemmaref{lem:induct}:
\begin{lemma}\label{lem:inductmult}
Let $C$ be a configuration and let $\mathcal{P}$ be a set of 
  $n \geq 2$ processes.
If $\mathcal{P}$ is bivalent from $C$, then there is a $\mathcal{P}$-only
  execution $\alpha$ and a set 
  $\mathcal{Q} \subseteq \mathcal{P}$ of at mostwot two processes such that $\mathcal{Q}$ is bivalent 
  from $C\alpha$ and there exists a $2\ell$-packing $\pi$ of the remaining 
  processes $\mathcal{P} - \mathcal{Q}$ in $C\alpha$.
\end{lemma}
\begin{fullversion}
\begin{proof}
By induction on $|\mathcal{P}|$. 
The base case is when $|\mathcal{P}| = 2$.
Let $\mathcal{Q} = \mathcal{P}$ and
let $\alpha$ be the empty execution.
Since $\mathcal{P} - \mathcal{Q}= \emptyset$,  the claim holds.
 
Now let $|\mathcal{P}| > 2$ and suppose the claim holds for 
  $|\mathcal{P}| - 1$. 
By~\lemmaref{lem:twosolo}, there exists a $\mathcal{P}$-only execution $\gamma$ 
starting from $C$
  and a set $\mathcal{Q} \subset \mathcal{P}$ of two processes  
  that is bivalent from $D = C\gamma$.
Pick any process $z \in \mathcal{P} - \mathcal{Q}$. 
Then $\mathcal{P} - \{z\}$ is bivalent from $D$ because $\mathcal{Q}$ is bivalent from $D$.

We construct a sequence of configurations $D_0, D_1, \ldots$
  reachable from $D$, such that, for all $i \geq 0$, 
  the following properties hold:
\begin{itemize}
\item[1.] there exists a set of two processes 
  $\mathcal{Q}_i \subseteq \mathcal{P} - \{z\}$ such that
  $\mathcal{Q}_i$ is bivalent from $D_i$,
\item[2.] there exists a $2\ell$-packing $\pi_i$ of the remaining processes 
  $\mathcal{R}_i = (\mathcal{P} - \{z\}) - \mathcal{Q}_i$ in $D_i$, and
\item[3.]
if $L_i$ is the set of all locations that are fully $2\ell$-packed by $\mathcal{R}_i$ in $D_i$,
then $D_{i+1}$ is reachable from $D_i$ by a 
  $(\mathcal{P}-\{z\})$-only execution $\alpha_i$
which contains a block multi-assignment $\beta_i$ such that, for each location $r \in L_i$,
there are at least $\ell$ multiple assignments in $\beta_i$ that perform $\ellbw$ on $r$.
\end{itemize}

By  the induction hypothesis applied to $D$ 
and $\mathcal{P}-\{z\}$, there is a $(\mathcal{P}-\{z\})$-only execution $\eta$ starting from $D$
and a set $\mathcal{Q}_0 \subseteq (\mathcal{P}-\{z\})$ of two processes such that
$\mathcal{Q}_0$ is
bivalent from $D_0 = D\eta$  and
and there exists a $2\ell$-packing $\pi_0$ of the remaining 
processes $\mathcal{R}_0 = (\mathcal{P} - \{z\}) - \mathcal{Q}_0$ in $D_0$.

Now suppose that $D_i$ is a configuration reachable from $D$ and $\mathcal{Q}_i$ and
  $\mathcal{R}_i$ are sets of processes that satisfy all three conditions.

By~\lemmaref{lem:multcover} applied to configuration $D_i$,
  there is a $\mathcal{Q}_i$-only execution $\xi_i$ such that 
  $\mathcal{R}_i \cup \mathcal{Q}_i = \mathcal{P} - \{z\}$
  is bivalent from $D_i\xi_i\beta_i$, where $\beta_i$ is a block 
  multi-assignment in which
  $\ellbw$ is performed exactly $\ell$ times on $r$, for each location $r \in L_i$.
Applying the induction hypothesis to $D_i\xi_i\beta_i$ 
  and $\mathcal{P} - \{z\}$, we get a $(\mathcal{P} - \{z\})$-only 
  execution $\psi_i$ leading to a configuration 
  $D_{i+1} = D_i\xi_i\beta_i\psi_i$, in which there is 
  a set, $\mathcal{Q}_{i+1}$, of two processes such that 
  $\mathcal{Q}_{i+1}$ is bivalent from $D_{i+1}$.
Additionally, there exists a $2 \ell$-packing $\pi_{i+1}$ of the remaining processes 
  $\mathcal{R}_{i+1} = (\mathcal{P} - \{z\}) - \mathcal{Q}_{i+1}$ in $D_{i+1}$.
Note that the execution $\alpha_i = \xi_i \beta_i \psi_i$ contains
  the block multi-assignment $\beta_i$.

Since there are only finitely many locations, 
  there exists $0 \leq i < j$ such that $L_i = L_j$, i.e.,
the set of fully $2\ell$-packed locations by $\mathcal{R}_i$ in $D_i$ is 
  the same as the set of fully $2\ell$-packed locations by $\mathcal{R}_j$ in $D_j$. 
Next, we  insert steps of $z$ that cannot be detected by any process in 
  $\mathcal{P} - \{z\}$. 
Consider any $\{z\}$-only execution $\zeta$ from $D_i\xi_i$ that decides 
  a value $v \in \{0,1\}$. 
If $\zeta$ does not contain any $\ellbw$ to locations outside $L_i$, then
$D_i\xi_i\zeta\beta_i$ is indistinguishable 
  from $D_i\xi_i\beta_i$ to processes in $\mathcal{P} - \{z\}$.
Since $D_i\xi_i\beta_i$ is bivalent for $\mathcal{P} - \{z\}$,
there exists a  $\mathcal{P} - \{z\}$ execution from $D_i\xi_i\beta_i$
and, hence, from $D_i\xi_i\zeta\beta_i$
that decides  $\bar{v}$, contradicting agreement.
Thus
  $\zeta$  contains an $\ellbw$ to a location not in $L_i$.
Let $\zeta'$ be the longest prefix of $\zeta$ that does not contain
an  $\ellbw$ to a location outside $L_i$.
Then,  in $D_i\xi_i\zeta'$, $z$ is poised to perform a multiple assignment
containing  an $\ellbw$ to a location outside $L_i = L_j$.

$D_i\xi_i\zeta'\beta_i$ is indistinguishable from 
  $D_i\xi_i\beta_i$ to $\mathcal{P} - \{z\}$, 
  so the $(\mathcal{P} - \{z\})$-only execution 
  $\psi_i \alpha_{i+1} \cdots \alpha_{j-1}$ can be applied from
$D_i\xi_i\zeta'\beta_i$. 
Let $\alpha = \gamma\eta\alpha_0 \cdots \alpha_{i-1} \xi_i\zeta'\beta_i\psi_i\alpha_{i+1} \cdots \alpha_{j-1}$. 
Every process in $\mathcal{P} - \{z\}$ is in the same state in $C\alpha$ 
  as it is in $D_j$. 
In particular, $\mathcal{Q}_j \subseteq \mathcal{P} - \{z\}$ 
  is bivalent from $D_j$ and, hence, from $C\alpha$. 
The $2\ell$-packing $\pi_j$ of $\mathcal{R}_j$ in $D_j$ is 
  a $2\ell$-packing of $\mathcal{R}_j$ in $C\alpha$ and $L_i = L_j$ is the set of locations 
  that are fully $2\ell$-packed by $\mathcal{R}_j$ in $C\alpha$. 
Since $z$ takes no steps after $D_i\xi\zeta'$,
$z$ covers a location $r$ not in $L_j$ in configurations $D_j$ and  $C\alpha$. 
Since $r \not \in L_j$, there is a $2\ell$-packing $\pi_j'$ 
  of $\mathcal{R}_j$ in $C\alpha$ which packs less than $2\ell$ processes into $r$. 
Let $\pi$ be the packing that packs $z$ into location $r$ and packs
  each process in $\mathcal{R}_j$ in the same location as $\pi'_j$ does.
Then $\pi$ is a
$2\ell$-packing of $\mathcal{R}_j \cup \{z\} = \mathcal{P} - \mathcal{Q}_j$
 in $C\alpha$.
\end{proof}
\end{fullversion}
We can now prove the main theorem.
\begin{theorem}
\label{thm:noconsmult}
Consider a memory consisting of $\ell$-buffers, in which each
process can atomically perform $\ellbw$ to any subset of the 
  $\ell$-buffers. 
Then any obstruction-free binary consensus algorithm for $n$ processes uses 
at least $\lceil (n-1)/2\ell \rceil$ locations.
\end{theorem}
\begin{proof}
Consider any obstruction-free binary consensus algorithm for $n$ processes.
By~\lemmaref{lem:initbi}, there exists
an initial configuration from which the set of all $n$ processes,
  $\mathcal{P}$, is bivalent. 
\lemmaref{lem:inductmult} implies that there is a  configuration, $C$, reachable from this initial configuration, 
a set of two processes $\mathcal{Q} \subseteq \mathcal{P}$
such that $\mathcal{Q}$ is bivalent from $C$,
and a $2\ell$-packing $\pi$ of the remaining processes
$\mathcal{R} = \mathcal{P} - \mathcal{Q}$ in $C$.
By the pigeonhole principle, $\mathcal{R}$ covers at least  $\lceil (n-2)/2\ell \rceil$ different locations.

Suppose that  $\mathcal{R}$ covers exactly $\lceil (n-2)/2\ell\rceil$ different locations
and $\lceil (n-2)/2\ell \rceil < \lceil (n-1)/2\ell\rceil$.
Then $n-2$ is a multiple of $2\ell$ and every location is fully
 $2\ell$-packed by $\mathcal{R}$.
Since $\mathcal{Q}$ is bivalent from $C$, \lemmaref{lem:multcover} implies that there is a $\mathcal{Q}$-only
execution $\xi$ such that some process in $\mathcal{Q}$ covers a location 
that is not fully $2\ell$-packed by $\mathcal{R}$. 
Hence, there are at least $\lceil (n-2)/2\ell\rceil + 1 = \lceil (n-1)/2\ell \rceil$  locations.
 \end{proof}

\section{Swap and Read}
\label{sec:swap}

In this section, we present an anonymous obstruction-free algorithm for solving $n$-consensus using $n-1$ shared memory locations, $X_1,\ldots,X_{n-1}$, which support $\id{read}$ and $\id{swap}$.
The $\id{swap(v)}$ instruction atomically
sets the memory location to have value $v$ and returns the value that it previously contained.

Intuitively, values $0, 1, \ldots, n-1$ are competing to complete \emph{laps}. If $v$ gets a substantial lead on all other values, then the value $v$ is decided. Each process has a local variable $\ell_v$,  for each $v \in \{0,1,\ldots,n-1\}$,
in which it stores its view of $v$'s current lap.
Initially, these are all 0.
If the process has input $v$, then its first step is to set $\ell_v = 1$.
The process also has $n$ local variables $a_1, \ldots, a_{n-1}$ and $s$.
In $a_i$,  it stores the last value it read from $X_i$, for $i \in \{1,\ldots,n-1\}$,
and, in $s$ it stores the  value returned by its last $\id{swap}$ operation.

When a process performs $\id{swap}$, it includes its process identifier and a strictly increasing sequence number
as part of its argument. Then it is possible to implement a linearizable, obstruction-free $\id{scan}$ of
the $n-1$ shared memory locations
using the double collect algorithm [AAD+93]:  A process repeatedly collects the values in all the locations (using $read$)
until it observes two consecutive collects with the same values.
In addition to a process identifier and a sequence number (which we will henceforth ignore),
each shared memory location stores a vector of $n$ components, all of which are initially 0.
Likewise, $a_1, \ldots, a_{n-1}$ and $s$ are initially $n$-component vectors of 0's.

A process begins by performing a $\id{scan}$ of all $n-1$ memory locations.
Then, for each value $v$, it updates its view, $\ell_v$, of $v$'s current lap
to be the maximum among $\ell_v$, the $v$'th component of $s$,
and the $v$'th component of the vector in each memory location when its last $\id{scan}$ was performed.
If there is a memory location that does not contain $(\ell_0,\ell_1,\ldots,\ell_{n-1})$, then the process performs
$\id{swap}((\ell_0,\ell_1,\ldots,\ell_{n-1}))$ on the first such location.
Now suppose all the memory locations contain  $(\ell_0,\ell_1,\ldots,\ell_{n-1})$.
If there is a value $v$ 
such that $\ell_v$ is at least 2 bigger than every other component in this vector,
then the process decides $v$. Otherwise, it picks the value $v$ with the largest current lap (breaking ties in favour of smaller values) and considers $v$ to have completed lap $\ell_v$. 
Then it performs
$\id{swap}((\ell_0,\ldots,\ell_{v-1},\ell_v+1,\ell_{v+1},\ldots,\ell_{n-1}))$
on  $X_1$. 
If the process doesn't decide, it repeats this sequence of steps.

\begin{algorithm} [!ht]
	\begin{algorithmic} [1]
		\State $\ell_x \gets 1$
				\Loop
		\State $(a_1,\ldots,a_{n-1}) \gets \mathit{scan}(X_1,\ldots,X_{n-1})$
		\For{$v \in \{0,1,\ldots,n-1\}$}
		\State \label{line:lapnum} $\ell_v \gets \max (\{\ell_v, s[v]\}  \cup \{ a_j[v] : 0 \leq j \leq n-1 \} )$ 
		\EndFor
		\State $\ell^* \gets \max \{ \ell_0,\ldots, \ell_{n-1}\}$
		\State $v^* \gets \min \{ v : \ell_v  = \ell^* \}$
		\If {$a_j = (\ell_0,\ldots,\ell_{n-1})$ for all $1 \leq j \leq n-1$} \Comment{$v^*$ has completed lap $\ell^*$}
		\If {$\ell^* \geq \ell_{v} + 2$ for all $v \neq v^*$} \Comment {$v^*$ is at least 2 laps ahead of all other values}
		\State \textbf{decide} $v^*$ and \textbf{terminate} 
		\EndIf
		\State $\ell_{v^*} \gets \ell_{v^*} + 1$ \Comment {value $v^*$ is now on the next lap}
		\EndIf
		\State $j \gets \min \{ j : a_j \neq (\ell_0,\ldots,\ell_{n-1}) \}$
		\State $s \gets \textit{swap}(X_j,(\ell_0,\ldots,\ell_{n-1}))$
		\EndLoop
	\end{algorithmic}
\caption{An $n$-consensus algorithm for a process with input value $x$}
	\label{consensus-racing-algorithm}
\end{algorithm}

We now prove that our protocol is correct, i.e.~it satisfies validity, agreement, and obstruction-freedom. Each step in the execution is either a $\id{swap}$ performed on some $X_j$ or a $\id{scan}$ of $X_1,\ldots,X_{n-1}$. 
For each $\id{scan}$, $S$, by a process $p$ and for each $v \in \{0,\ldots,n-1\}$, we define $\ell_v(S)$ to be the value of $p$'s local variable $\ell_v$ computed on line~\ref{line:lapnum} following $S$. Similarly, for each $\id{swap}$  $U$ and each $v \in \{0,\ldots,n-1\}$, if $U$ swaps the contents of $X_j$ with value $(\ell_0,\ldots,\ell_{n-1})$, then we define $\ell_v(U) = \ell_v$.

We begin with an easy observation, which follows from inspection of the code. 

\begin{observation}
	\label{race-algorithm-observation}
	Let $U$ be a $\id{swap}$ by some process $p$ and let $S$ be the last $\id{scan}$ that $p$ performed before $U$. Then, for each $v \in \{0,\ldots,n-1\}$, $\ell_v(U) \geq \ell_v(S)$. If there exists $v \in \{0,\ldots,n-1\}$ such that $\ell_v(U) > \ell_v(S)$, then $\ell_v(U) = \ell_v(S)+1$, $\ell_{v'}(S) \leq \ell_{v}(S)$  for all $v' \neq v$, and $S$ returned
the same value, $(\ell_0(S),\ldots,\ell_{n-1}(S))$, from each shared memory location.
\end{observation}

The next lemma follows from Observation \ref{race-algorithm-observation}. It says that if there was a $\id{scan}$, $S$, where value $v$ is on lap $\ell > 0$, i.e. $\ell_v(S) = \ell$, then there was a $\id{scan}$ where $v$ is on lap $\ell-1$ and all the swap objects contained this information.

\begin{lemma}
	\label{race-algorithm-incremental}
	Let $S$ be any $\id{scan}$ and let $v \in \{0,\ldots,n-1\}$. If $\ell_v(S) > 0$, then there was a $\id{scan}$, $T$,
	prior to $S$
such that $T$ returned
the same value from each shared memory location,
$\ell_v(T) = \ell_v(S) - 1$, and $\ell_{v'}(T) \leq \ell_v(T)$, for all $v' \neq v$.
\end{lemma}
\begin{proof}
	Since each swap object initially contains an $n$-component vector of 0's
and $\ell_v(S) > 0$, there was $\id{swap}$ $U$ prior to $S$ such that $\ell_v(U) = \ell_v(S)$. Consider the first such $\id{swap}$. Let $p$ be the process that performed $U$ and let $T$ be the last $\id{scan}$ performed by $p$ before $U$. By the first part of Observation \ref{race-algorithm-observation}, $\ell_v(U) \geq \ell_v(T)$ and, by definition of $U$, $\ell_v(U) > \ell_v(T)$ (otherwise, there would have been an earlier $\id{swap}$ $U'$ with $\ell_v(U') = \ell_v(T) = \ell_v(U) = \ell_v(S)$). By the second part of Observation \ref{race-algorithm-observation}, it follows that $\ell_v(U) = \ell_v(T) + 1$, $\ell_{v'}(T) \leq \ell_{v}(T)$ for all $v' \neq v$, and $T$ returns a vector whose components all contain the same pair $(\ell_0(T),\ldots,\ell_{n-1}(T))$. Since $\ell_v(U) = \ell_v(S)$, it follows that $\ell_v(T) = \ell_v(U)-1 = \ell_v(S) - 1$.
\end{proof}

The following lemma is key. In particular, it says that if a process considers value $v$ to have completed  lap $\ell$ as a result of performing a $\id{scan}$ $S$ where all the components have the same value, then every process will think that $v$ is at least on lap $\ell$ when it performs any $\id{scan}$ after $S$.

\begin{lemma} \label{race-algorithm-key-lemma}
Suppose $S$ is a $\id{scan}$ that returns the same value from each shared memory location.
	 If $T$ is a $\id{scan}$ performed after $S$, then, for each $v \in \{0,\ldots,n-1\}$, 
	$\ell_v(T) \geq \ell_v(S)$.
\end{lemma}
\begin{proof}
	Suppose, for a contradiction, that there is a \id{scan} $T$ after $S$ such that $\ell_v(T) < \ell_v(S)$ for some $v \in \{0,\ldots,n-1\}$. Consider the first such \id{scan} $T$.
	Since the the value of $\ell_v$ computed on line~\ref{line:lapnum} is the maximum of a set that includes $\ell_v$,
the value of each local variable $\ell_v$ is non-decreasing.
Since $\ell_v(T) < \ell_v(S)$, the process, $q$, that performed $T$ is not the process that performed $S$. 
	
The value returned by $T$ from each shared memory location is either the value returned by $S$ from that location
or is the argument of a $\id{swap}$ performed on that location between $S$ and $T$.
By assumption, the value returned by $S$ from each shared memory location is $(\ell_0(S),\ldots,\ell_{n-1}(S))$.
Since $\ell_v(T) < \ell_v(S)$ and the value of $\ell_v$ computed by $q$ on line~\ref{line:lapnum} after $T$
is at least as large as the $v$'th components
of the values returned by $T$ from each shared memory location, 
it follows that the value returned by $T$ from each shared memory location
is the argument of a $\id{swap}$ performed on that location between $S$ and $T$.

Partition the \id{swap}s that occur between $S$ and $T$ into two sets, $W$ and $W'$.
For any $\id{swap}$ $Y$ performed by a process $p$  between $S$ and $T$,  $Y \in W$ if $p$'s last \id{scan} prior to $Y$ occurred before $S$.
Otherwise $Y \in W'$.
In particular, if process $p$ performed $S$, then all of the $\id{swap}$s between $S$ and $T$ performed by $p$ are in $W'$.

Each process alternately performs $\id{scan}$ and $\id{swap}$.
Therefore, if a process performs more than one $\id{swap}$ between $S$ and $T$,
then all of them, except possibly the first, are in $W'$.
It follows that each $\id{swap}$ in $W$ is by a different process and
$|W| \leq n-1$.

If the \id{swap}s in $W$ modify fewer than $n-1$ different swap objects,
then the value returned by $T$ from some shared memory location is the argument of a $\id{swap}$ $U' \in W'$.
such that $\ell_v(U') < \ell_v(S)$.

Otherwise, the \id{swap}s in $W$ modify exactly $n-1$ shared memory locations.
Then the process, $q$, that performed $T$ performed a $\id{swap}$ $U \in W$.
Each $\id{swap}$ in $W$ modifies a different location. Therefore,
the value that $U$ returns is either the value returned by $S$ from that memory location
or is the argument of a $\id{swap}$ $U' \in W'$ (performed at the same location).
Since $\ell_v(T) < \ell_v(S)$ and the value of $\ell_v$ computed by $q$ on line~\ref{line:lapnum} after $T$
is at least as large as the $v$'th component 
of the result of every $\id{swap}$ that $q$ performed prior to $T$,
it follows that $U$ returns the argument of a $\id{swap}$ $U' \in W'$
such that $\ell_v(U') < \ell_v(S)$.

In either case, let $p'$ be the process that performed $U'$ and let $T'$ be the last $\id{scan}$
that $p'$ performed prior to $U'$.
By definition of $W'$, $T'$ occurs between $S$ and $U'$ and, hence, before $T$. By definition of $T$,  $\ell_v(T') \geq \ell_v(S)$, for each $v \in \{0,1,\dots,n-1\}$. Therefore, by Observation~\ref{race-algorithm-observation}, $\ell_v(U') \geq \ell_v(T')$, so  $\ell_v(U') \geq \ell_v(S)$. This is a contradiction.
\end{proof}

The previous lemma allows us to prove that once a value $v$ is at a lap $\ell$ that is 2 laps ahead of $\overline{v}$ and every swap object contains this information, then $\overline{v}$ will never reach lap $\ell$, i.e.~$v$ will always be at least one lap ahead of $\overline{v}$.

\begin{lemma}
	\label{race-algorithm-stay-ahead}
Suppose $S$ is a $\id{scan}$ that returns the same value $(\ell_0(S),\ldots,\ell_{n-1}(S))$ from each shared memory location and there is some $v \in \{0,\ldots,n-1\}$ such that $\ell_v(S) \geq \ell_{v'}(S) + 2$ for all $v' \neq v$.
Then, for every $\id{scan}$ $T$ and every value $v' \neq v$,  $\ell_{v'}(T) \leq \ell_{v'}(S)+1$.
\end{lemma}
\begin{proof}
Suppose, for a contradiction, that there is some $\id{scan}$ $T$ and some value $v' \neq v$
such that $\ell_{v'}(T) \geq \ell_{v'}(S)+2 >0 $.
Consider the first such $\id{scan}$.
By Lemma \ref{race-algorithm-incremental}, there was a $\id{scan}$, $T'$, prior to $T$
such that $T'$ returned the same value from each shared memory location,
$\ell_{v'}(T') = \ell_{v'}(T)-1$, and $\ell_{v}(T') \leq \ell_{v'}(T')$. 
By definition of $T$, $\ell_{v'}(T') < \ell_{v'}(S) + 2$. Hence, $\ell_{v'}(T) = \ell_{v'}(S) + 2$ and $\ell_{v'}(T') = \ell_{v'}(S) + 1$.
	
Since $\ell_{v'}(S) < \ell_{v'}(T')$ and $T'$ returned the same value from each shared memory location,
the contrapositive of Lemma \ref{race-algorithm-key-lemma}
implies that $S$ was performed before $T'$.
Since $S$ returned the same value from each shared memory location, Lemma \ref{race-algorithm-key-lemma}
implies that $\ell_v(T') \geq \ell_v(S)$. By assumption, $\ell_v(S) \geq \ell_{v'}(S) + 2$. Hence, $\ell_v(T') \geq \ell_{v'}(S)+2 = \ell_{v'}(T') + 1$. This contradicts the fact that $\ell_{v}(T') \leq \ell_{v'}(T')$.
\end{proof}

We can now prove that the protocol satisfies agreement, validity, and obstruction-free termination.

\begin{lemma}
	No two processes decide differently.
	\label{race-algorithm-agreement}
\end{lemma}
\begin{proof}
From lines 8--10 of the code, the last step a process performs before deciding value $v^*$ is
a  $\id{scan}$, $S$, that returns the same value from each shared memory location and
such that $\ell_{v^*}(S) \geq \ell_{v}(S)+2$ for all $v \neq v^*$.
Consider the first such $\id{scan}$.
By Lemma \ref{race-algorithm-key-lemma}, $\ell_{v^*}(T) \geq \ell_{v^*}(S)$ for every $\id{scan}$ $T$ performed after $S$.
By Lemma \ref{race-algorithm-stay-ahead},
$\ell_{v}(T) \leq \ell_{v}(S)+1$ for all $v \neq v^*$.
 Hence, $\ell_{v^*}(T) > \ell_{v}(T)$. It follows that no process ever decides $v \neq v^*$.
\end{proof}

\begin{lemma}
	If every process has input $x$, then no process decides $x' \neq x$.
	\label{race-algorithm-validity}
\end{lemma}
\begin{proof}
Suppose there is a $\id{swap}$, $U$, such that $\ell_{x'}(U) > 0$ for some $x' \neq x$.
Consider the first such $\id{swap}$.  Let $p$ be the process that performed $U$ and
let $S$ be the last $\id{scan}$ that $p$ performed before $U$.
Since each shared memory location initially stores an $n$-component vector of 0's, 
$\ell_{x'}(S) = 0 < \ell_{x'}(U)$.
By the second part of Observation \ref{race-algorithm-observation},
$\ell_x(S) \leq \ell_{x'}(S) = 0$.
Since $p$ has input $x$, it set $\ell_x = 1$ on line 1.
From the code, $\ell_x$ is non-decreasing, so $\ell_x \geq 1$ whenever $p$ performs line 5.
By definition, $\ell_{x}(S) \geq 1$.
This is a contradiction.
Thus, $\ell_{x'} (U) = 0$ for all $\id{swap}$s $U$.
Since no process has input $x'$, no process set $\ell_{x'} =1$ on line 1.
It follows that $\ell_{x'}(S) = 0$ for all $\id{scan}$s $S$, so no process decides $x' \neq x$.
\end{proof}

\begin{lemma}
	Every process decides after performing at most $3n-2$ $\id{scan}$s in a solo execution.
	\label{race-algorithm-obstruction-free}
\end{lemma}
\begin{proof}
	Let $p$ be any process and consider the first $\id{scan}$ $S$ performed by $p$ in its solo execution. After performing at most $n-1$ $\id{swap}$s, all with value $(\ell_0(S),\ldots,\ell_{n-1}(S))$, $p$ will perform a $\id{scan}$ 
that returns $(\ell_0(S),\ldots,\ell_{n-1}(S))$ from each shared memory location. 	
Let $v^* = \min \{v : \ell_v(S) \geq \ell_{v'}(S) \textrm{ for all } v' \neq v \}$. If $\ell_{v*}(S) \geq \ell_{v}(S) + 2$ for all $v \neq v^*$, then $p$ decides $v^*$. Otherwise, $p$ performs $n-1$ $\id{swap}$s, all with value $(\ell_0',\ldots,\ell_{n-1}')$, where $\ell_{v^*}' = \ell_{v^*}(S) + 1$ and $\ell_{v}' = \ell_{v}(S)$, for $v \neq v^*$. Then it performs a $\id{scan}$ that returns a vector whose components all contain $(\ell_0',\ldots,\ell_{n-1}')$.
	If $\ell_{v^*}' \geq \ell_{v}' + 2$ for all $v \neq v^*$, then $p$ decides $v^*$.
	If not, then $p$ performs an additional $n-1$ $\id{swap}$s, all with value $(\ell_0'',\ldots,\ell_{n-1}'')$, where $\ell_{v^*}'' = \ell_{v^*}' + 1 = \ell_{v^*}(S)+2$ and $\ell_{v}'' = \ell_{v}' = \ell_{v}(S)$ for $v \neq v^*$. Finally, $p$ performs a $\id{scan}$ that returns a vector whose components all contain $(\ell_0'',\ldots,\ell_{n-1}'')$ and decides $v^*$. Since $p$ performs at most $3(n-1)$ $\id{swap}$s and each $\id{swap}$ is immediately followed by a $\id{scan}$, this amounts to $3n-2$ $\id{scan}$s, including the first scan, $S$.
\end{proof}

The preceding lemmas immediately yield the following theorem.

\begin{theorem}
	\label{race-algorithm-main-theorem}
	There is an anonymous, obstruction-free protocol for solving consensus among $n$ processes that uses only $n-1$ memory locations supporting read and swap.
\end{theorem}

In~\cite{FHS98}, there is a proof that $\Omega(\sqrt{n})$ shared memory locations are necessary 
to solve obstruction-free consensus
when the system only supports $\id{swap}$ and $\id{read}$ instructions.
\section{Test-and-Set and Read}
\label{sec:tas}

Consider a system that supports only
$\id{test-and-set}()$ and $\id{read}()$. If there are only $2$ processes,
then it is possible to solve wait-free consensus using a single memory
location. However, we claim that any algorithm for solving obstruction-free 
binary consensus among $n \geq 3$ processes must use an unbounded number of 
memory locations. The key is to prove the following analogue of Lemma~\ref{lem:cover}.

\begin{lemma}
Let $\mathcal{P}$ be a set of at least 3 processes and let $C$ be a configuration. If $\mathcal{P}$ is bivalent from $C$, then, for every $k \geq 0$, there exists a $\mathcal{P}$-only execution $\alpha_k$ from $C$ such that $\mathcal{P}$ is bivalent from $C\alpha_k$ and at least $k$ locations have been set to 1 in $C\alpha_k$.
\label{lem:tas}
\end{lemma}
\begin{proof}
By induction on $k$. The base case, $k = 0$, holds when $\alpha_0$ is the  the empty execution.
Given $\alpha_k$, for some $k \geq 0$, we show how to construct $\alpha_{k+1}$.
By Lemma \ref{lem:twosolo}, there is a $\mathcal{P}$-only execution $\xi$ from
$C\alpha_k$ and two processes, $p_0, p_1 \in \mathcal{P}$
such that $p_i$ decides $i$ in its terminating solo execution, $\gamma_i$,
from $C\alpha_k\xi$, for $i = 0,1$.
Let $L_k$ be the set of at least $k$ memory locations that have been set to 1 in
$C\alpha_k\xi$.

Let $z \in \mathcal{P}-\{p_0,p_1\}$.
Suppose that $z$ decides $v \in \{0,1\}$ in its solo execution $\delta$ from $C\alpha_k\xi$.
If $z$ does not perform $\id{test-and-set}()$ on a location outside $L_k$ during $\delta$,
then $C\alpha_k\xi\delta$ is indistinguishable from $C\alpha_k\xi$ to $\{p_0,p_1\}$,
so $\gamma_{\bar{v}}$ can be applied starting from $C\alpha_k\xi\delta$,
violating agreement.
Thus $z$ performs $\id{test-and-set}()$ on a location outside $L_k$ during $\delta$.
Let $\beta$ be the shortest prefix of $\delta$ in which $z$ performs $\id{test-and-set}()$ on a location outside $L_k$
and let $r$ be the location outside $L_k$ on which $z$ performs
$\id{test-and-set}()$ during $\beta$.

If $\{p_0,p_1\}$
is bivalent from $C\alpha_k\xi\beta$, then $\alpha_{k+1} = \alpha_k\xi\beta$
satisfies the claim for $k+1$, since $\beta$ sets location $r \not\in L_k$ to 1.
So, without loss of
generality, suppose that $\{p_0,p_1\}$ is  0-univalent from $C\alpha_k\xi\beta$.
Let $\psi$ be the longest prefix of $\gamma_1$ such
that $p_0$ decides 0 from $C\alpha_k\xi\psi\beta$.
Note that $\psi \neq \gamma_1$, since 1 is decided in $\gamma_1$.
Let $\psi'$ be the first step in $\gamma_1$ following $\psi$.
If $\psi'$ is a $\id{read}$ or a $\id{test-and-set}()$ on a location in $L_k \cup \{r\}$,
then $C\alpha_k\xi\psi\psi'\beta$ is indistinguishable
from $C\alpha_k\xi\psi\beta$ to $p_0$. This is impossible,
since $p_0$ decides 0 from $C\alpha_k\xi\psi\beta$
and $p_0$ decides 1 from $C\alpha_k\xi\psi\psi'\beta$.
Thus, $\psi'$ is
a $\id{test-and-set}()$ on a location outside $L_k \cup \{r\}$.
Hence $C\alpha_k\xi\psi\psi'\beta$ is indistinguishable
from $C\alpha_k\xi\psi\beta\psi'$ to all processes.
In particular, $p_0$ decides 1 from $C\alpha_k\xi\psi\beta\psi'$.

Since $p_0$ decides 0 from $C\alpha_k\xi\psi\beta$
and decides 1 from $C\alpha_k\xi\psi\beta\psi'$,
it follows that
$\{p_0,p_1\}$ is bivalent from $C\alpha_k\xi\psi\beta$.
Furthermore, $\beta$ sets location $r \not\in L_k$ to 1.
Thus $\alpha_{k+1} = \alpha_k\xi\psi\beta$
satisfies the claim for $k+1$.
\end{proof}

By~\lemmaref{lem:initbi}, there is an initial configuration from which the set 
of all processes in the system is bivalent. Then it follows
from~\lemmaref{lem:tas}
that any binary consensus algorithm for 
$n \geq 3$ processes uses an unbounded number of locations.

\begin{theorem}
	For $n \geq 3$, it is not possible to solve $n$-consensus using a 
	bounded number of memory locations supporting only $\id{read}()$ and 
	$\id{test-and-set}()$.
\end{theorem}

There is an algorithm for obstruction-free binary consensus that uses an 
unbounded number of shared memory locations that support only $\id{read}()$ and
$\id{write}(1)$~\cite{GR05}.
All locations are initially 0.
The idea is to simulate a counter using 
an unbounded number of binary registers and then to run the racing counters 
algorithm presented in~\lemmaref{lem:cntr}. 
In this algorithm, there are two unbounded tracks on which 
processes race, one for preference 0 and one for preference 1.
Each track consists of an unbounded sequence of shared memory locations.
To indicate progress, a process
performs $\id{write}(1)$ to the location on its preferred track
from which it last read 0.
Since the count on each track does not decrease, a process can 
perform a $\id{scan}$ using the double collect algorithm~\cite{AADGMS93}.
It is not necessary to read all the locations in a track to determine
the count it represents. It suffices to read from the location on the track from which
it last read 0, continuing to read from the subsequent locations on the track until it reads another 0.
A process changes its preference if it sees
that the number of 1's on its preferred track is less than the number of 1's on the
other track. Once a process sees that its preferred track  is at least 2 ahead of the
other track, it decides its current preference.

It is possible to generalize this 
algorithm to solve $n$-valued consensus by having $n$ tracks, each consisting of
an unbounded sequence of shared memory locations.
Since $\id{test-and-set}()$ can simulate
$\id{write}(1)$ by ignoring the value returned, we get 
the following result.

\begin{theorem}
It is possible to solve $n$-consensus using an 
unbounded number of memory locations supporting only $\id{read}()$ and either 
$\id{write}(1)$ 
or $\id{test-and-set}()$.
\end{theorem}

Now, suppose we can also perform $\id{write}$(0) or $\id{reset}()$ a memory location from 1 to 0. 
There is an existing binary consensus algorithm that uses $2n$
locations, each storing a single bit~\cite{B11}.
Then, it is possible to solve $n$-consensus using
$O(n\log n)$ locations by applying~\lemmaref{lem:bitbybit}. 
There is a slight subtlety, since the algorithm in the proof of~\lemmaref{lem:bitbybit}
uses two designated locations for each round,
to which values in $\{0,\ldots,n-1\}$ can be written.
In place of each designated location, it is possible to use a sequence of $n$ binary locations, all initialized to 0.
Instead of performing $\id{write}(x)$ on the designated location, 
  a process performs $\id{write}(1)$ to the $(x+1)$'st binary location. 
To find one of the values that has been written to the designated location, 
  a process $\id{read}$s the sequence of binary locations until it sees a $1$.

\begin{theorem}
	It is possible to solve $n$-consensus using $O(n\log n)$ memory 
	locations supporting only $\id{read}()$, either $\id{write}(1)$ or 
	$\id{test-and-set}()$, and either $\id{write}(0)$ or $\id{reset}()$.
\end{theorem}

\section{Conclusions and Future Work}
\label{sec:future} 

In this paper, we defined a hierarchy based on the space complexity of 
  solving obstruction-free consensus.
We used consensus because it is a well-studied, 
  general problem that seems to capture a fundamental difficulty 
  of multiprocessor synchronization. 
Moreover, consensus is {\em universal}: any sequentially defined 
  object can be implemented in a wait-free way using only consensus 
  objects and registers~\cite{Her91}. 

We did not address the issue of universality within our hierarchy.
One history object can be used to implement 
  any sequentially defined object.
Consequently, it may make sense to consider defining a hierarchy 
  on sets of instructions based on implementing 
  a history object, a compare-and-swap object, 
  or a repeated consensus object shared by $n$ processes.
However, the number of locations required for solving $n$-consensus is  
  the same as the number of locations required for obstruction-free
  implementations of these long-lived objects 
  for many of the instruction sets that we considered. 

A truly accurate complexity-based hierarchy
  would have to take step complexity into consideration. 
Exploring this may be an important future direction.
Also, it is standard to assume that memory locations have unbounded size,
  in order to focus solely on the challenges of synchronization.
For a hierarchy to be truly practical, however, we might need
  to consider the size of the locations used by an algorithm.

There are several other interesting open problems.
To the best of our knowledge, all existing space lower bounds
  rely on a combination of covering and indistinguishability arguments.
However, when the covering processes apply\ $\id{swap}(x)$, as opposed 
  to $\id{write}(x)$, they can observe  differences between executions, 
  so they can no longer be reused and still maintain indistinguishability.
This means that getting a larger space lower bound for 
  $\{ \id{swap}(x), \id{read}() \}$ would most likely require new techniques.
An algorithm that uses less than $n-2$ shared memory locations
would be even more 
  surprising, as the processes would have to
  modify the sequence of memory locations they access based on
  the values they receive from $\id{swap}$s,
  to circumvent the argument from~\cite{Zhu16}. 
The authors are unaware of any such algorithm.

Getting an $\omega(\sqrt{n})$ space lower bound for
  solving consensus in a system that supports $\id{test-and-set}()$,
  $\id{reset}()$ and $\id{read}()$ is also interesting.
Using $\id{test-and-set}()$, processes can observe difference 
  between executions as they can  using $\id{swap}(x)$.
However, each location can only store a single bit.
This restriction could potentially help in proving a lower bound.

To prove the space lower bound of $\lceil \frac{n-1}{\ell} \rceil$
  for $\ell$-buffers, we extended the technique of~\cite{Zhu16}.
The $n-1$ lower bound of~\cite{Zhu16} has since been improved
  to $n$ by~\cite{EGZ18}.
Hence, we expect that the new simulation-based technique used there
  can also be extended to prove
  a tight space lower bound of $\lceil \frac{n}{\ell} \rceil$.

We conjecture that, for sets of instructions, ${\cal I}$, which 
  contain only $\id{read}()$, $\id{write}(x)$, and either 
  $\id{increment}()$ or $\id{fetch-and-increment}()$, 
  ${\cal SP}({\cal I},n) \in \Theta(\log n)$. 
Similarly, we conjecture, for 
  $\mathcal{I} = \{\id{read}(),\id{write}(0),\id{write}(1)\}$, 
  $\mathcal{SP}(\mathcal{I},n) \in \Theta(n\log n)$.
Proving these conjectures is likely to require techniques that depend 
  on the number of input values, such as in the lower bound 
  for $m$-valued adopt-commit objects in~\cite{AspE14}.

We would like to understand the properties of sets of instructions 
  at certain levels in the hierarchy.
For instance, what properties enable a collection of instructions to 
  solve $n$-consensus using a single location?
Is there an interesting characterization of the sets of instructions 
  ${\cal I}$ for which ${\cal SP}({\cal I},n)$ is constant?
How do subsets of a set of instructions relate to one another 
  in terms of their locations in the hierarchy? 
Alternatively, what combinations of sets of instructions decrease 
 the amount of space needed to solve consensus?
For example, using only $\id{read}()$, $\id{write}(x)$, and either 
  $\id{increment}()$ or $\id{decrement}()$, more than one memory 
  location is needed to solve binary consensus. 
But with both $\id{increment}()$ and $\id{decrement}()$, a single 
  location suffices.
Are there general properties governing these relationships?

\section{Acknowledgments}
Support is gratefully acknowledged from
the Natural Science and Engineering Research Council of Canada,
the National Science Foundation under grants CCF-1217921, 
        CCF-1301926, and IIS-1447786, the Department of Energy under grant ER26116/DE-SC0008923, 
        and Oracle and Intel corporations.

The authors would like to thank Michael Coulombe, Dan Alistarh, Yehuda Afek, Eli Gafni
  and Philipp Woelfel for helpful conversations and feedback.

\bibliographystyle{alpha}
\bibliography{biblio}

\newcommand{\etalchar}[1]{$^{#1}$}
\begin{thebibliography}{GHHW13}

\bibitem[AAC09]{AAC09}
James Aspnes, Hagit Attiya, and Keren Censor.
\newblock Max registers, counters, and monotone circuits.
\newblock In {\em Proceedings of the 28th ACM Symposium on Principles of
  Distributed Computing}, PODC '09, pages 36--45, 2009.

\bibitem[AAD{\etalchar{+}}93]{AADGMS93}
Yehuda Afek, Hagit Attiya, Danny Dolev, Eli Gafni, Michael Merritt, and Nir
  Shavit.
\newblock Atomic snapshots of shared memory.
\newblock {\em Journal of the ACM}, 40(4):873--890, 1993.

\bibitem[AE14]{AspE14}
James Aspnes and Faith Ellen.
\newblock Tight bounds for adopt-commit objects.
\newblock {\em Theory of Computing Systems}, 55(3):451--474, 2014.

\bibitem[AH90]{AH90}
James Aspnes and Maurice Herlihy.
\newblock Fast randomized consensus using shared memory.
\newblock {\em Journal of Algorithms}, 11(3):441--461, 1990.

\bibitem[AW04]{AW04}
Hagit Attiya and Jennifer Welch.
\newblock {\em Distributed computing: fundamentals, simulations, and advanced
  topics}, volume~19.
\newblock John Wiley \& Sons, 2004.

\bibitem[Bow11]{B11}
Jack~R. Bowman.
\newblock Obstruction-free snapshot, obstruction-free consensus, and
  fetch-and-add modulo k.
\newblock Technical Report TR2011-681, Computer Science Department, Dartmouth
  College, 2011.
\newblock http://www.cs.dartmouth.edu/reports/TR2011-681.pdf.

\bibitem[BRS15]{BRS15}
Zohir Bouzid, Michel Raynal, and Pierre Sutra.
\newblock Anonymous obstruction-free (n, k)-set agreement with n-k+1 atomic
  read/write registers.
\newblock {\em Distributed Computing}, pages 1--19, 2015.

\bibitem[EGSZ16]{EGSZ16}
Faith Ellen, Rati Gelashvili, Nir Shavit, and Leqi Zhu.
\newblock A complexity-based hierarchy for multiprocessor
  synchronization:[extended abstract].
\newblock In {\em Proceedings of the 35th ACM Symposium on Principles of
  Distributed Computing}, PODC '16, pages 289--298, 2016.

\bibitem[EGZ18]{EGZ18}
Faith Ellen, Rati Gelashvili, and Leqi Zhu.
\newblock Revisionist simulations: A new approach to proving space lower
  bounds.
\newblock In {\em Proceedings of the 37th ACM symposium on Principles of
  Distributed Computing}, PODC '18, 2018.

\bibitem[FHS98]{FHS98}
Faith~Ellen Fich, Maurice Herlihy, and Nir Shavit.
\newblock On the space complexity of randomized synchronization.
\newblock {\em Journal of the ACM}, 45(5):843--862, 1998.

\bibitem[FLMS05]{FLMS05}
Faith~Ellen Fich, Victor Luchangco, Mark Moir, and Nir Shavit.
\newblock Obstruction-free algorithms can be practically wait-free.
\newblock In {\em Proceedings of the 19th International Symposium on
  Distributed Computing}, DISC '05, pages 78--92, 2005.

\bibitem[Gel15]{Gel15}
Rati Gelashvili.
\newblock On the optimal space complexity of consensus for anonymous processes.
\newblock In {\em Proceedings of the 29th International Symposium on
  Distributed Computing}, DISC '15, pages 452--466, 2015.

\bibitem[GHHW13]{GHHW13}
George Giakkoupis, Maryam Helmi, Lisa Higham, and Philipp Woelfel.
\newblock An $\mathcal{O}(\sqrt{n})$ space bound for obstruction-free leader
  election.
\newblock In {\em Proceedings of the 27th International Symposium on
  Distributed Computing}, DISC '13, pages 46--60, 2013.

\bibitem[GR05]{GR05}
Rachid Guerraoui and Eric Ruppert.
\newblock What can be implemented anonymously?
\newblock In {\em Proceedings of the 19th International Symposium on
  Distributed Computing}, DISC '05, pages 244--259, 2005.

\bibitem[Her91]{Her91}
Maurice Herlihy.
\newblock Wait-free synchronization.
\newblock {\em ACM Transactions on Programming Languages and Systems},
  13(1):124--149, 1991.

\bibitem[HR00]{HR00}
Maurice Herlihy and Eric Ruppert.
\newblock On the existence of booster types.
\newblock In {\em Proceedings of the 41st {IEEE} Symposium on Foundations of
  Computer Science}, FOCS '00, pages 653--663, 2000.

\bibitem[HS12]{HS12Book}
Maurice Herlihy and Nir Shavit.
\newblock {\em The Art of Multiprocessor Programming}.
\newblock Morgan Kaufmann, 2012.

\bibitem[Int12]{intel}
Intel.
\newblock {\em Transactional Synchronization in {Haswell}}, 2012.
\newblock
  http://software.intel.com/en-us/blogs/2012/02/07/transactional-synchronization-in-haswell.

\bibitem[Jay93]{Jay93}
Prasad Jayanti.
\newblock On the robustness of herlihy's hierarchy.
\newblock In {\em Proceedings of the 12th ACM Symposium on Principles of
  Distributed Computing}, PODC '93, pages 145--157, 1993.

\bibitem[LH00]{hl00}
Wai-Kau Lo and Vassos Hadzilacos.
\newblock All of us are smarter than any of us: Nondeterministic wait-free
  hierarchies are not robust.
\newblock {\em SIAM Journal on Computing}, 30(3):689--728, 2000.

\bibitem[MPR18]{MPR18}
Achour Most{\'e}faoui, Matthieu Perrin, and Michel Raynal.
\newblock A simple object that spans the whole consensus hierarchy.
\newblock {\em arXiv preprint arXiv:1802.00678}, 2018.

\bibitem[Ray12]{Ray12Book}
Michel Raynal.
\newblock {\em Concurrent programming: algorithms, principles, and
  foundations}.
\newblock Springer Science \& Business Media, 2012.

\bibitem[Rup00]{Rup00}
Eric Ruppert.
\newblock Determining consensus numbers.
\newblock {\em SIAM Journal on Computing}, 30(4):1156--1168, 2000.

\bibitem[Sch97]{Sch97}
Eric Schenk.
\newblock The consensus hierarchy is not robust.
\newblock In {\em Proceedings of the 16th {ACM} Symposium on Principles of
  Distributed Computing}, PODC '97, page 279, 1997.

\bibitem[Tau06]{Tau06Book}
Gadi Taubenfeld.
\newblock {\em Synchronization algorithms and concurrent programming}.
\newblock Pearson Education, 2006.

\bibitem[Zhu15]{Zhu15}
Leqi Zhu.
\newblock Brief announcement: Tight space bounds for memoryless anonymous
  consensus.
\newblock In {\em Proceedings of the 29th International Symposium on
  Distributed Computing}, DISC '15, page 665, 2015.

\bibitem[Zhu16]{Zhu16}
Leqi Zhu.
\newblock A tight space bound for consensus.
\newblock In {\em Proceedings of the 48th {ACM} Symposium on Theory of
  Computing}, STOC '16, pages 345--350, 2016.

\end{thebibliography}

\end{document}